\documentclass{fundam}

\publyear{22}
\papernumber{2102}
\volume{185}
\issue{1}

\usepackage{url} 
\usepackage[ruled,lined,linesnumbered]{algorithm2e}
\usepackage{graphicx}

\usepackage{amssymb}

\usepackage[bookmarks=false]{hyperref}
\usepackage[capitalize, nameinlink, noabbrev]{cleveref}

\AtBeginEnvironment{lemma}{\crefalias{definition}{lemma}}
\AtBeginEnvironment{theorem}{\crefalias{definition}{theorem}}

\usepackage{macros}
\usepackage{subcaption}
\usepackage{multirow}
\usepackage{shuffle}

\definecolor{darkgreen}{rgb}{0.0, 0.5, 0.0}

\begin{document}

\title{Hierarchical Decomposition of Separable Workflow-Nets}


\author{Humam Kourani\corresponding\\
Fraunhofer Institute for Applied Information Technology FIT \\ 
Schloss Birlinghoven, 53757 Sankt Augustin, Germany\\
humam.kourani{@}fit.fraunhofer.de
\and Gyunam Park\\
Fraunhofer Institute for Applied Information Technology FIT  \\
Schloss Birlinghoven, 53757 Sankt Augustin, Germany
\and Wil M.P. van der Aalst\\
RWTH Aachen University \\
Ahornstraße 55, 52074 Aachen, Germany
} 

\maketitle

\runninghead{H. Kourani et al.}{Hierarchical Decomposition of Separable Workflow-Nets}

\begin{abstract}
   The Partially Ordered Workflow Language (POWL) has recently emerged as a process modeling notation, offering strong quality guarantees and high expressiveness. While early versions of POWL relied on strict block-structured operators for choices and loops, the language has recently evolved into POWL 2.0, introducing \emph{choice graphs} to enable the modeling of non-block-structured decisions and cycles. To bridge the gap between the theoretical advantages of POWL and the practical need for compatibility with established notations, robust model transformations are required. This paper presents a novel algorithm for transforming safe and sound workflow nets (WF-nets) into equivalent POWL 2.0 models. The algorithm recursively identifies structural patterns within the WF-net and translates them into their POWL representation. Unlike the previous approach that required separate detection strategies for exclusive choices and loops, our new algorithm utilizes choice graphs to capture generalized decision and cyclic patterns. We formally prove the correctness of our approach, showing that the generated POWL model preserves the language of the input WF-net. Furthermore, we prove the completeness of our algorithm on the class of \emph{separable WF-nets}, which corresponds to nets constructed via the hierarchical nesting of \emph{state machines} and \emph{marked graphs}. We evaluate our algorithm on large-scale process models to demonstrate its high scalability. Furthermore, to test its practical expressiveness, we applied it to a benchmark of 1,493 industrial and synthetic process models. Our algorithm successfully transformed all models in this benchmark, suggesting that POWL 2.0's expressive power, although a formal subclass of sound workflow nets, is generally sufficient to capture the complex logic found in real-world business processes. This work paves the way for broader adoption of POWL in practical process analysis and improvement applications.
\end{abstract}


\section{Introduction}\label{sec:intro}
The field of process modeling and analysis relies heavily on formal notations to represent and reason about the behavior of complex systems. While standard notations such as Petri nets \cite{DBLP:journals/pieee/Murata89}, and specifically workflow nets (WF-nets) \cite{DBLP:conf/apn/Aalst97}, and Business Process Model and Notation (BPMN) \cite{DBLP:books/el/15/RosingWCM15} have gained widespread adoption, they suffer from limitations in terms of their ability to guarantee desirable quality properties such as \textit{soundness} (i.e., the absence of deadlocks and other anomalies).

Hierarchical representations have a longer history in process modeling. In particular, \emph{WF-trees} were introduced as Refined Process Structure Trees (RPSTs) enriched with behavioral annotations for sound free-choice WF-systems \cite{DBLP:conf/apn/WeidlichPMW10,DBLP:journals/fuin/WeidlichPMW11}. These WF-trees are closely related to the operator-based \emph{process tree} representations commonly used in process mining \cite{DBLP:series/lnbip/Leemans22}. POWL \cite{DBLP:conf/bpm/KouraniZ23} builds on the latter modeling paradigm by retaining a recursively defined hierarchical syntax while introducing more expressive composition constructs. The hierarchical nature of these languages offers several advantages. First, these languages guarantee soundness by construction. Furthermore, the structured representation can significantly enhance the understandability of complex process models for humans, making it easier to grasp the overall control flow and identify potential areas for improvement. Finally, hierarchical modeling languages open up opportunities for developing faster and more efficient techniques for different process mining tasks. For example, previous work has demonstrated the advantages of POWL in process discovery from data \cite{DBLP:journals/is/KouraniZSA25} and process modeling from text \cite{DBLP:conf/bpmds/KouraniB0A24}.

Process trees construct models by recursively combining submodels using a set of block-structured operators, such as sequence ($\to$), exclusive choice ($\xor$), parallel execution ($\wedge$), and loop ($\Loop$). The Partially Ordered Workflow Language (POWL) \cite{DBLP:conf/bpm/KouraniZ23} was introduced as a generalization of process trees that preserves desirable properties while increasing expressiveness. POWL replaces the parallel operator with \emph{partial orders}. This allows for the representation of complex concurrency dependencies that do not fit into rigid blocks. The language has been further extended (\emph{POWL 2.0} \cite{DBLP:conf/bpm/KouraniPA25}) to replace the $\xor$ and $\Loop$ operators with \emph{choice graphs}, which allow for modeling arbitrary exclusive paths and cycles within a single hierarchical node. By combining partial orders (for generalized concurrency) and choice graphs (for generalized decisions and cycles), POWL 2.0 preserves the desirable properties of process trees while increasing expressiveness through more generalized structures. 

\cref{fig:ex} shows an example POWL 2.0 model (\cref{fig:ex:powl}) for a retailer’s order fulfillment process and a WF-net (\cref{fig:ex:wf}) that captures the same behavior.
The production subprocess follows non-block-structured concurrency where executing production depends on both gathering the required materials and scheduling, while notifying the customer is only constrained by scheduling. Simultaneously, the top-level control flow exhibits non-block-structured decision logic, encompassing the initial choice between in-stock and production paths, the choice between cancellation or shipment for in-stock orders, and the choice between ending the process or looping back to the initial state after a cancellation.



\begin{figure}[!t]
    \centering    
        \begin{subfigure}{\textwidth}
            \centering
            \includegraphics[width=\textwidth]{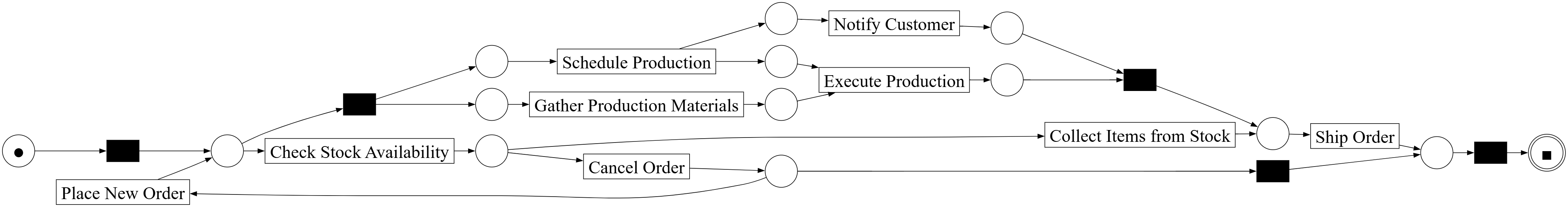}
            \caption{A WF-net.}\label{fig:ex:wf}
        \end{subfigure}
    
        \begin{subfigure}{\textwidth}
            \centering
            \includegraphics[width=0.8\textwidth]{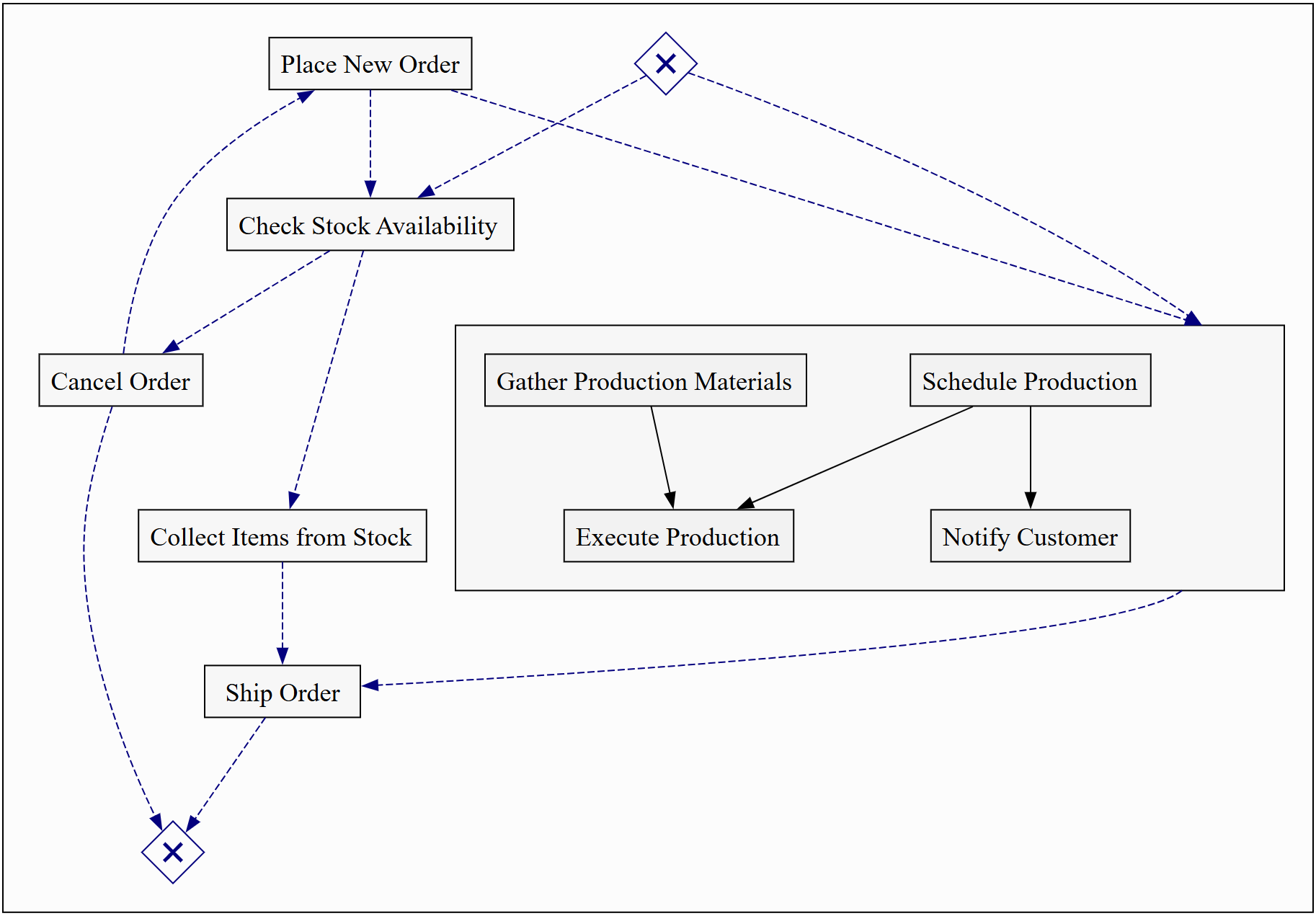}
            \caption{Modeling the same behavior in POWL 2.0. The top-level control flow is modeled using a \emph{choice graph}, which contains a nested \emph{partial order} for the production subprocess.}\label{fig:ex:powl}
        \end{subfigure}

        \caption{Example models for a retailer’s order fulfillment process. \label{fig:ex}}
\end{figure}

Despite its demonstrated benefits, the practical adoption of POWL is challenged by the prevalence of standard notations in existing tools and workflows. While optimized analysis and improvement functionalities have been developed for POWL (e.g., the automated process redesign feature in ProMoAI \cite{DBLP:conf/ijcai/KouraniB0A24}), they remain inaccessible to processes modeled using standard notations such as WF-nets. Furthermore, advanced process mining techniques can be adapted to fully leverage POWL's hierarchical nature (e.g., efficient conformance checking \cite{DBLP:conf/bpm/RochaA23}). This motivates the need for a robust conversion from standard notations to POWL, allowing us to harness the benefits of POWL without requiring the adjustment of existing tools and practices.

While the conversion from POWL to a sound WF-net is relatively straightforward \cite{DBLP:journals/is/KouraniZSA25}, the inverse transformation presents a more significant challenge. This asymmetry arises from the fact that POWL represents a subclass of sound WF-nets. In previous work \cite{DBLP:conf/apn/KouraniPA25}, we proposed an algorithm to translate WF-nets into POWL. The proposed algorithm recursively decomposes the input WF-net into its smaller parts, identifies XOR, Loop, and partial order patterns, and assembles them into an equivalent POWL model. However, that approach was designed for the initial variant of POWL and inherited its limitations: it cannot handle WF-nets with complex, unstructured decision and cyclic logic.

This paper proposes a generalized algorithm that translates sound WF-nets into POWL 2.0. By leveraging the unified structure of choice graphs, our new approach simplifies the transformation process: rather than searching for disparate choice and loop blocks, we recursively decompose the input WF-net into structural patterns corresponding to either choice graphs (covering all exclusive/cyclic behavior) or partial orders (covering all concurrent behavior). We formally prove the correctness of this transformation, demonstrating that the generated POWL 2.0 model captures the language of the original WF-net. Furthermore, we formally characterize the class of \emph{separable WF-nets}, workflow nets constructed by recursively substituting transitions with subnets that behave locally as either \emph{state machines} (capturing generalized choice and cyclic logic) or \emph{marked graphs} (capturing generalized concurrent logic). We prove that our algorithm is complete for this class, meaning it successfully transforms any safe and sound separable WF-net into its equivalent POWL 2.0 representation.

To illustrate the expressive power of POWL 2.0 and the effectiveness of our transformation, consider the retailer's order fulfillment process shown in \cref{fig:ex:wf}. This WF-net captures a process characterized by non-block-structured concurrency and decision logic. As shown in \cref{fig:ex:powl}, our proposed algorithm successfully encapsulates the complex top-level routing logic into a choice graph, while modeling the internal production dependencies through a nested partial order. 

The effectiveness of our approach is evaluated through a two-fold analysis. First, we evaluate the robustness of the algorithm using a collection of 1,493 industrial and synthetic workflow nets. Given that POWL 2.0 represents a subclass of sound WF-nets, this assessment allows us to measure the practical coverage of the language, specifically, whether its expressive power is sufficient to represent the diverse logic found in real-world business processes. Second, we assess scalability on large-scale models to ensure that the recursive decomposition remains efficient as process size increases. Our results demonstrate the effectiveness of our approach at successfully capturing the complex structures found in practice while maintaining high performance.  This work represents a significant step towards enabling the development of process analysis and improvement techniques that utilize POWL internally while seamlessly accepting inputs in widely used formats.

The remainder of the paper is structured as follows. \cref{sec:rel} discusses related work. \cref{sec:pre} introduces necessary preliminaries. In \cref{sec:convert}, we detail our algorithm for translating WF-nets into POWL 2.0. \cref{sec:gua} provides formal proofs for the algorithm's correctness and completeness. Finally, we evaluate our approach in \cref{sec:eval} and conclude in \cref{sec:conc}.

\section{Related Work}\label{sec:rel}

\subsection{Process Model Transformations}

Transformations between process modeling notations have been explored in various contexts. Some research focuses on transformations between different Petri net classes, such as the work on unfolding Colored Petri Nets into standard Place/Transition nets in \cite{DBLP:conf/apn/Dal-Zilio20} and the work on reducing free-choice Petri nets to either T-nets (also called marked graphs) or P-nets (also called state machines) in \cite{DBLP:conf/apn/Aalst21}. Other research addresses transformations between different types of process models. For example, \cite{DBLP:journals/infsof/DijkmanDO08} proposes an approach for converting BPMN models into Petri nets, \cite{gardner2003uml} discusses translating UML diagrams into BPEL, and \cite{DBLP:conf/apn/LangnerSW98} explores the mapping of Event-driven Process Chains (EPCs) into colored Petri nets. An overview on different approaches for the translation between workflow graphs and free-choice workflow nets is provided in \cite{DBLP:journals/is/FavreFV15}.

Particularly relevant to this work are approaches for transforming WF-nets into hierarchical or block-structured process models. The work on translating WF-nets to BPEL in \cite{DBLP:journals/infsof/AalstL08,DBLP:conf/otm/LassenA06} employs a bottom-up strategy, iteratively identifying patterns corresponding to BPEL fragments and substituting each identified pattern with a single transition to continue the recursion. This approach aims to maximize the size of the detected components in each iteration. The approach presented in \cite{DBLP:journals/algorithms/ZelstL20} for translating WF-nets into process trees, while also employing a bottom-up strategy, restricts the search space to patterns of size two. This approach cannot be adapted to POWL due to the presence of advanced constructs that cannot be decomposed into components of size two.

The fundamental difference between our algorithm and the aforementioned approaches, besides the choice of the target language, lies in the search strategy. While existing methods rely on bottom-up pattern matching and the iterative substitution of local fragments, our approach employs a top-down decomposition strategy. It starts with the entire WF-net and recursively partitions it into smaller sub-nets by identifying high-level structural patterns. By analyzing the global structure of the net rather than performing an exhaustive search for small, localized components, we aim at achieving higher scalability on large-scale models.

\subsection{Structural Decomposition and Process Structuring}

A second line of work related to our approach concerns the structural decomposition of workflow models. The Refined Process Structure Tree (RPST) decomposes a workflow graph into a canonical hierarchy of Single-Entry-Single-Exit (SESE) fragments and can be computed in linear time \cite{DBLP:conf/wsfm/PolyvyanyyVV10}. Based on the RPST, Weidlich et al.\ introduce the \emph{WF-tree}, which enriches the decomposition with behavioral information for sound free-choice WF-systems \cite{DBLP:conf/apn/WeidlichPMW10,DBLP:journals/fuin/WeidlichPMW11}. Their analysis also considers rigid fragments that are S-nets or T-nets, which correspond closely to the state-machine and marked-graph structures underlying the two main constructs of POWL 2.0. 

Further work addresses the structuring of rigid fragments. Polyvyanyy et al.\ propose a maximal-structuring approach for acyclic process models based on the RPST, unfolding, ordering relations graphs, and modular decomposition \cite{DBLP:journals/cj/PolyvyanyyGFW14,DBLP:phd/dnb/Polyvyanyy12}. Ordering relations graphs capture causal, conflict, and concurrent relations and can expose additional structure inside rigid RPST components. This work is closely related to POWL 2.0, but has a different scope. It focuses on structuring acyclic process models, whereas POWL 2.0 is itself an executable hierarchical modeling language whose choice graphs may also contain cycles. Moreover, ordering relations graphs may capture acyclic components in which conflict and concurrency are intertwined, while POWL 2.0 separates these forms of behavior through hierarchical nesting. Hence, the two formalisms are related, but neither is assumed to subsume the other.

Structural decomposition has also been studied as a basis for compositional analysis of workflow nets. Polyvyanyy et al.\ use graph connectivity to decompose WF-nets stepwise and exploit the resulting components for reasoning about safeness and soundness \cite{DBLP:journals/acta/PolyvyanyyWW11}. Their results demonstrate how structural decomposition can support local reasoning about behavioral correctness. This work is relevant to the structural guarantees established in this paper, but serves a different objective: the decomposition in \cite{DBLP:journals/acta/PolyvyanyyWW11} is derived from connectivity properties for verification purposes, whereas our projections are induced by the structural requirements of the target POWL 2.0 constructs.

\section{Preliminaries}\label{sec:pre}
This section introduces fundamental preliminaries and notations.

\subsection{Basic Notations}\label{sec:pre:notation}

A \emph{multi-set} generalizes the notion of a set by tracking the frequencies of its elements. A multi-set over a set $X$ is expressed as $M = [{x_{1}}^{c_1},..., {x_{n}}^{c_n}]$ where $x_{1},..., x_{n} \in X$ are the elements of $M$ (denoted as $x_{i} \in M$ for $1\leq i\leq n$) and $M(x_i) = c_i \geq 1$ is the frequency of $x_{i}$ for $1\leq i\leq n$. 

A sequence of length $n \geq 0$ over a set $X$ is defined as function $\sigma{\colon}\{1,...,n\} \to X$, and we express it as $\sigma = \langle \sigma(1), ..., \sigma(n)\rangle$. The set of all finite sequences over $X$ is denoted by $X^{*}$. The concatenation of two sequences $\sigma_1$ and $\sigma_2$ is expressed as $\sigma_1 \cdot \sigma_2$, e.g., $\langle x_{1} \rangle \cdot \langle x_{2}, x_{1} \rangle = \langle x_{1}, x_{2}, x_{1} \rangle$. For two sets of sequences $L_1$ and $L_2$, we write $L_1 \cdot L_2 = \{\sigma_1 \cdot \sigma_2 \ | \ \sigma_1\in L_1 \ \wedge \ \sigma_2\in L_2\}$. 

Let ${\po \subseteq X\times X}$ be a binary relation over a set $X$. We use ${x_1 \po x_2}$ to denote ${(x_1,x_2)\in \po}$ and ${x_1 \notpo x_2}$ to denote ${(x_1,x_2) \notin \po}$. We define the \emph{transitive closure} of $\po$ as $\closure\po = \{(x, y) \mid \exists_{x_1, \dots, x_n \in X} \ x = x_1 \ \wedge \ y = x_n \ \wedge \forall_{1\leq i<n} \; x_i \po x_{i+1}\}$.

A \emph{strict partial order} (\emph{partial order} for short) over a set $X$ is a binary relation that is \emph{irreflexive} ($x \notpo x$ for all ${x\in \X}$) and \emph{transitive} (${x_1 \po  x_2} \wedge {x_2 \po  x_3} \Rightarrow {x_1 \po  x_3}$). Irreflexivity and transitivity imply \emph{asymmetry} (${x_1 \po  x_2} \Rightarrow {x_2 \notpo  x_1}$). 



For a non-empty set $X$, a \emph{partition} of $X$ of size $n \geq 1$ is a set of subsets $P = \{X_1, ... , X_n\}$ such that $X = X_1 \cup ... \cup X_n$, $X_i \neq \emptyset$ for every $1 \leq i \leq n$, and $X_i \cap X_j = \emptyset$ for every $1 \leq i < j \leq n$. For any $x  \in X$, we write $P_x$ to denote the subset of the partition (also called \emph{part}) that contains $x$, i.e., $P_x \in \{X_1, ... , X_n\}$ and $x \in P_x$. For example, let $P = \{\{a, b\}, \{c\}\}$ be a partition of $\{a, b, c\}$ of size $2$. Then, $P_a = P_b = \{a, b\}$ and $P_c = \{c\}$. 

Let $f: X \rightarrow Y$ be a function. $f$ is called \emph{injective} if no different elements of $X$ are mapped into the same element of $Y$; i.e., $\forall_{x,x' \in X, \; x\neq x'} f(x) \neq f(x')$. $f$ is called \emph{surjective} if it covers all elements of $Y$; i.e., $\forall_{y \in Y} \exists_{x \in X} f(x) = y$. $f$ is called \emph{bijective} if it is both injective and surjective. For two sets $X$ and $Y$, $\setBIs(X, Y) = \{f: X \rightarrow Y \ | \ \text{f is bijective}\}$ denotes the set of all bijective functions from $X$ to $Y$.

\subsection{Workflow Nets}

We use $\ActUniverse$ to denote the set of all activities. We use $\tau \notin \ActUniverse$ to denote the \emph{silent activity}, which is used to model a choice between executing or skipping a path in process model, for example. To enable creating process models with duplicated activities, we introduce the notion of \emph{transitions}, and we use $\TraUniverse$ to denote the set of all transitions. Each transition is mapped to an activity, denoted as the $label$ of the transition. We use $\lab{\colon} \TraUniverse \to \ActUniverse \cup \{\tau\}$ to denote the labeling function.

A \textit{Petri net} is a directed bipartite graph consisting of two types of nodes, \textit{places} and \textit{transitions}; we follow the standard terminology and semantics from \cite{DBLP:journals/pieee/Murata89}. Transitions represent instances of activities, while places are used to model dependencies between transitions.

\begin{definition}[Petri Net]\label{def:petrinet}
A Petri net is a triple $N = (P, T, F)$, where $T \subset \TraUniverse$ is a finite set of transitions, $P$ is a finite set of places such that $T \cap P = \emptyset$, and $F \subseteq (P \times T) \cup (T \times P)$ is the flow relation.  
\end{definition}

Let $N=(P,T, F)$ be a Petri net. We define the following notations: 
\begin{itemize}
    \item The \emph{transition reachability relation} $\TTR \subseteq T \times T$ is defined as follows for $t, t' \in T$: $t \TTR t' \iff (t, t') \in \closure{F}$. 
    \item For $x \in P \cup T$, $\pre{x} = \{y \ | \ (y, x) \in F\}$ is the \emph{pre-set} of $x$, and $\post{x} = \{y \ | \ (x, y) \in F\}$ is the \emph{post-set} of $x$.
    \item For $T' \subseteq T$, we define the \emph{projection} of $P$ on $T'$ as $P\project{T'} = \{ p \in P \ | \ (\pre{p} \cup \post{p}) \cap T' \neq \emptyset\}$.
    \item For $P' \subseteq P$ and $T' \subseteq T$, we define the \emph{projection} of $F$ on $P$ and $T$ as $F\project{P',T'} = F \cap ((P' \times T') \cup (T' \times P'))$.
\end{itemize}

To reason about the structural equivalence of process models and to ensure the termination of our recursive decomposition, we introduce the formal notion of \emph{isomorphism} between Petri nets. Isomorphism captures the idea that two Petri nets are identical in their structure and labeling, differing only in the identities of their internal places and transitions.

\begin{definition}[Petri Net Isomorphism]\label{def:isomorphism}
Let $N_1 = (P_1, T_1, F_1)$ and $N_2 = (P_2, T_2, F_2)$ be two Petri nets. $N_1$ and $N_2$ are \emph{isomorphic}, denoted $N_1 \cong N_2$, if there exist two bijections $f_P \colon P_1 \to P_2$ and $f_T \colon T_1 \to T_2$ such that:
\begin{enumerate}
    \item $\forall p \in P_1, t \in T_1 \colon (p, t) \in F_1 \iff (f_P(p), f_T(t)) \in F_2$,
    \item $\forall t \in T_1, p \in P_1 \colon (t, p) \in F_1 \iff (f_T(t), f_P(p)) \in F_2$,
    \item $\forall t \in T_1 \colon \lab(t) = \lab(f_T(t))$.
\end{enumerate}
\end{definition}

Places hold \emph{tokens}, and a transition is considered \emph{enabled} if each of its preceding places has at least one token. \emph{Firing} an enabled transition consumes one token from each of its preceding places and produces a token in each of its succeeding places. A \emph{marking} is a multi-set of places indicating the number of tokens in each place. A Petri net is called \emph{safe} if each place in the net cannot hold more than one token.

In the context of process modeling, we focus on a specific subclass of Petri nets designed to represent processes with a well-defined beginning and end, known as \emph{workflow nets} \cite{DBLP:conf/apn/Aalst97}.

\begin{definition}[Workflow Net]\label{def:wf-net}
Let $N=(P,T, F)$ be a Petri net. $N$ is a workflow net (WF-net) iff places $N_{source}, N_{sink}\in P$ exist such that:
    \begin{itemize}
        \item \textbf{Unique source:} $\{N_{source}\} = \{p \in P \ | \ \pre{p} = \emptyset\}$.
        \item \textbf{Unique sink:} $\{N_{sink}\} = \{p \in P \ | \ \post{p} = \emptyset\}$.
        \item \textbf{Connectivity:} each node is on a directed path from $N_{source}$ to $N_{sink}$; equivalently, every node is reachable from $N_{source}$ and can reach $N_{sink}$.
    \end{itemize}
\end{definition}

For a workflow net $N$, we use $N_{source}$ to denote the unique source place and $N_{sink}$ to denote the unique sink place. 

WF-nets provide the structural basis for modeling business processes. However, to ensure that the behavior of the net is predictable and well-suited for structured analysis, further constraints on the interaction between choice and synchronization may be required. The \emph{free-choice} property is a widely used condition that ensures that if a transition is enabled, the choice to fire it depends only on the state of its input places, rather than on synchronization constraints with other parallel branches.

\begin{definition}[Free-Choiceness]\label{def:free-choice}
A Petri net $N=(P,T, F)$ is \emph{free-choice} iff for any two transitions $t_1, t_2 \in T$: 
\[
(\pre{t_1} \cap \pre{t_2} \neq \emptyset) \Rightarrow (\pre{t_1} = \pre{t_2}).
\]
\end{definition}

WF-nets may suffer from quality anomalies (e.g., transitions that can never be enabled). WF-nets without such undesirable properties are called \emph{sound}.

\begin{definition}[Soundness]\label{def:sound}
Let $N=(P,T, F)$ be a WF-net. $N$ is sound iff the following conditions hold:
\begin{itemize}
    \item \textbf{No dead transitions:} for each transition $t\in T$, there exists a marking $M$ reachable from $[N_{source}]$ that enables $t$.
    \item \textbf{Option to complete:} for every marking $M$ reachable from $[N_{source}]$, there exists a firing sequence leading from $M$ to $[N_{sink}]$. 
    \item \textbf{Proper completion:} $[N_{sink}]$ is the only marking reachable from $[N_{source}]$ with at least one token in $N_{sink}$. 
\end{itemize}
\end{definition}

The WF-net shown in \cref{fig:ex:wf} is sound. Note that the option to complete implies proper completion.

\subsection{POWL Language}\label{sec:powl}
A POWL model is a hierarchical process representation constructed recursively from a set of activities. In its original formulation \cite{DBLP:conf/bpm/KouraniZ23}, the language relied on combining submodels either as partial orders or using the control flow operators $\xor$ and $\Loop$. The operator $\xor$ models a decision between alternative branches, while $\Loop$ captures cyclic behavior between two submodels: the \textit{do-part} is executed first, and each time the \textit{redo-part} is executed, it is followed by another execution of the do-part. In a partial order, all submodels are executed, while respecting the given execution order.

POWL 2.0 generalizes the concepts of choice and loops into a unified structure called a \emph{choice graph}. A choice graph is a directed graph that models arbitrary exclusive paths through a set of submodels. By replacing the rigid $\xor$ and $\Loop$ operators with choice graphs, POWL 2.0 retains the hierarchical advantages of the original language while significantly expanding its expressiveness to cover non-block-structured workflow nets.

\subsubsection{Choice Graphs}
A choice graph consists of a set of nodes connected by directed edges such that each node lies on a path from a designated start node to a designated end node.

\begin{definition}[Choice Graph]\label{def:choice_graph}
A \emph{choice graph} over a set of nodes $X$ is a tuple $\cg = (N, E)$ where:
\begin{itemize}
    \item $N = X \cup \{\xorstart{N}, \xorend{N}\}$ with two artificial start and end nodes $\xorstart{N}, \xorend{N} \notin X$.
    \item $E \subseteq N \times N$ is a binary relation over $N$.
    \item $\xorstart{N}$ is the unique start node, i.e., $\{\xorstart{N}\} = \{x \in N \ | \ (N \times \{x\}) \cap E = \emptyset\}$.
    \item $\xorend{N}$ is the unique end node, i.e., $\{\xorend{N}\} = \{x \in N \ | \ ( \{x\} \times N) \cap E = \emptyset\}$.
    \item Every node is on a directed path from $\xorstart{N}$ to $\xorend{N}$; equivalently, every node is reachable from $\xorstart{N}$ and can reach $\xorend{N}$.
\end{itemize}
\end{definition}

A choice graph represents a collection of possible execution paths, each beginning at the start node and ending at the end node. Formally, for a choice graph $\cg = (N, E)$ over a set $X$, we define the set of all paths in $\cg$ as follows:
\[
\paths{\cg} = \{ \langle x_1, ..., x_n \rangle \in X^* \ | \ (\xorstart{N}, x_1), (x_1, x_2), ..., (x_{n-1}, x_n), (x_n, \xorend{N}) \in E\}.
\]

Note that the start and end nodes are structural delimiters and are not included in the execution sequence $X^*$.

\subsubsection{Syntax and Semantics}
We now present the formal syntax and semantics of POWL 2.0. While our definition aligns with the one introduced in \cite{DBLP:conf/bpm/KouraniPA25}, we define POWL models over transitions rather than activities. This modification is necessary to distinguish between multiple instances of the same activity label within a single process model, thereby allowing for the representation of duplicate activities.

\paragraph*{Notation.} To facilitate treating partial orders and choice graphs as abstract structures that can be applied to any set of nodes, we introduce the following notation. Let \( n \geq 2 \) and \( X = \{x_1, \dots, x_n\} \) be a set of \( n \) elements.
\begin{itemize}
    \item \( \Orders{n} \) denotes the set of all partial orders over \( \{1, \dots, n\} \).
    
    \item For a partial order \( \po \in \Orders{n} \), we use \( \po(x_1, \dots, x_n) \) to denote the partial order \( \po' \) over \( X \) defined as follows: $i \po j \Leftrightarrow x_i \po' x_j \text{ for all } i, j \in \{1, \dots, n\}$.
    
    \item \( \Graphs{n} \) denotes the set of all choice graphs over \( \{1, \dots, n\} \).
    
    \item For a choice graph \( \cg = (N, E) \in \Graphs{n} \), \( \cg(x_1, \dots, x_n) \) denotes the choice graph \( \cg' = (N', E') \) over \( X \) where $N' = X \cup \{\xorstart{}, \xorend{}\}$ and $E' = \{(f(u), f(v)) \mid (u, v) \in E\}$ with a mapping function $f$ defined as $f(i) = x_i$ for $i \in \{1, \dots, n\}$, $f(\xorstart{}) = \xorstart{}$, and $f(\xorend{}) = \xorend{}$.
\end{itemize}

\begin{definition}[POWL Model]\label{def:powl} 
POWL models are defined as follows:
\begin{itemize}
    \item Any transition ${t \in \TraUniverse}$ is a POWL model.
    \item Let $\powl_1, ..., \powl_n$ be $n \geq 2$ POWL models.
    \begin{itemize}
        \item For a partial order $\po \in \Orders{n}$, $\po(\powl_1, ..., \powl_n)$ is a POWL model.
        \item For a choice graph $\cg \in \Graphs{n}$, $\cg(\powl_1, ..., \powl_n)$ is a POWL model.
        
    \end{itemize}
\end{itemize}
\end{definition}

The semantics of a POWL model are defined by the set of possible execution sequences it generates, referred to as its \textit{language}. This language is defined recursively based on the structure of the POWL model. The language of a choice graph corresponds to the union of sequences generated by concatenating the languages of submodels along every valid path in the graph. Conversely, the language of a partial order corresponds to all possible interleavings that satisfy the specified partial order constraints.

Let \( \sigma_1, \dots, \sigma_n \in X^* \) be sequences over a set \( X \) with \( n \geq 2 \), and let \( \po \in \Orders{n} \). The \emph{order-preserving shuffle operator} \( \shuffle_{\po} \) produces the set of sequences obtained by interleaving \( \sigma_1, \dots, \sigma_n \) while maintaining both the partial order \( \po \) among the sequences and the inherent sequential order within each sequence. For example, consider the sequences \( \sigma_1 = \langle a, b\rangle \), \( \sigma_2 = \langle c\rangle \), \( \sigma_3 = \langle d, e\rangle \), and the partial order \( \po = \{(1, 2), (1, 3)\} \in \Orders{3} \). Then, the set of valid interleavings is $\shuffle_{\po}(\sigma_1, \sigma_2, \sigma_3) = \{\langle a, b, c, d, e\rangle, \langle a, b, d, c, e\rangle, \langle a, b, d, e, c\rangle\}$.

\begin{definition}[Order-Preserving Shuffle Operator]\label{def:shuffle}
Let \( \sigma_1, \dots, \sigma_n \in X^* \) be sequences over a set \( X \) with \( n \geq 2 \), and let \( \po \in \Orders{n} \). We define the set of all indexed positions in the input sequences as:
\[ I = \{ (j, k) \mid 1 \leq j \leq n \ \wedge \ \ 1 \leq k \leq |\sigma_j| \}. \]
Then the order-preserving shuffle operator \( \shuffle_{\po} \) is defined as follows: \\
$\shuffle_{\po}(\sigma_1, \dots, \sigma_n ) =   \{\sigma \in X^* \ | \ |\sigma| = |I| \wedge \  \mathop{\exists}\limits_{f \in \setBIs(I, \{1, \dots, |\sigma|\})} \, \Big($ \\ 
$\hspace*{1cm} \mathop{\forall}\limits_{(j,k) \in I} \sigma(f(j, k)) = \sigma_j(k)$ \\
$\hspace*{1cm}  \wedge \ \mathop{\forall}\limits_{(j_1,k_1), (j_2,k_2) \in I} (j_1 \po j_2 \vee (j_1=j_2 \wedge k_1 < k_2)) \Rightarrow f(j_1,k_1) < f(j_2,k_2) \Big) \}.$
\end{definition}

\begin{definition}[POWL 2.0 Semantics]\label{def:ext_sem}
The language of a POWL model is recursively defined as follows:

\begin{itemize}
    \item ${\lang(t) = \{\langle a \rangle\}}$ for $t \in \TraUniverse$ with $\lab(t) = a \in \ActUniverse$.
    \item ${\lang(t) = \{\langle \rangle\}}$ for $t \in \TraUniverse$ with $\lab(t) = \tau$.
    \item Let $\powl_1, ..., \powl_n$ be ${n \geq 2}$ POWL models.
    \begin{itemize}     
         \item For $\po \in \Orders{n}$, $\lang(\po(\powl_1, ..., \powl_n)) = \{\sigma \in \shuffle_{\po}(\sigma_1 , ..., \sigma_n) \ | \ \forall_{1 \leq i \leq n} \ \sigma_i \in \lang(\powl_i) \}$.   

         \item For $\cg \in \Graphs{n}$, $\lang(\cg(\powl_1, ..., \powl_n)) = \bigcup_{\langle i_1, ..., i_k \rangle \in \paths{\cg}} \lang(\powl_{i_1}) \cdot ... \cdot \lang(\powl_{i_k})$.
    \end{itemize}
\end{itemize}
\end{definition}

\cref{fig:ex:powl} shows an example POWL model. This model is defined by a top-level choice graph that captures the complex decision and cyclic logic of the process. Nested within this structure is a partial order submodel that enforces causal dependencies between the concurrent production tasks. This hierarchical nesting allows POWL to accurately represent non-block-structured behavior while maintaining soundness by construction.

\subsection{Scope of Translation: Separable WF-nets}\label{sec:separable_wf}
POWL 2.0 is highly expressive, yet it represents a specific subclass of sound workflow nets. To clearly define the scope of our translation algorithm and the class of nets for which we guarantee completeness, we analyze the structural correspondence between POWL 2.0 constructs and Petri nets. 

POWL 2.0 constructs map to two fundamental subclasses of Petri nets: \emph{state machines} and \emph{marked graphs}. A choice graph is structurally equivalent to a sound state machine fragment. In a state machine, every transition has at most one input and one output place, meaning all flow divergences are decision points. A partial order is structurally equivalent to a sound marked graph. In a marked graph, every place has at most one input and one output transition, meaning all flow divergences are parallel splits/joins.

\begin{definition}[State Machine]
Let $N=(P,T, F)$ be a Petri net. $N$ is a \emph{state machine (SM)} iff every transition $t \in T$: 
\[
|\pre{t}| \leq 1 \quad \text{and} \quad |\post{t}| \leq 1.
\]
\end{definition}

\begin{definition}[Marked Graph]
Let $N=(P,T, F)$ be a Petri net. $N$ is a \emph{marked graph (MG)} iff every place $p \in P$: 
\[
|\pre{p}| \leq 1 \quad \text{and} \quad |\post{p}| \leq 1.
\]
\end{definition}

More details on these classes and their properties can be found in established literature, such as \cite{DBLP:conf/ac/1996petri1,desel1995free}.

Since a POWL 2.0 model is a hierarchical composition of these constructs, the corresponding WF-net is a hierarchy of SESE (Single-Entry-Single-Exit) fragments, where each fragment behaves locally as either a state machine or a marked graph. To formalize the nesting of sub-processes, we first introduce the formal notion of \emph{substitutive composition}. This operation allows for the recursive substitution of a transition with an entire sub-net.

\begin{definition}[Substitutive Composition]
Let $N = (P, T, F)$ be a Petri net and $t \in T$. Let $N' = (P', T', F')$ be a WF-net such that $(P \cup T) \cap (P' \cup T') = \emptyset$. The \emph{substitutive composition} of $N$ with $N'$ at $t$, denoted as $N_{[t \rightarrow N']}$, is the Petri net $(P'', T'', F'')$ defined by: 
\begin{itemize}
    \item $P'' = P \cup (P' \setminus \{N'_{source}, N'_{sink}\})$,
    \item $T'' = (T \setminus \{t\}) \cup T'$,
    \item $F'' = \{(u, v) \in F \mid u \neq t \wedge v \neq t\} \cup \{(u, v) \in F' \mid u \neq N'_{source} \wedge v \neq N'_{sink} \}$ \\
    $\cup \{(p, t') \mid (p, t) \in F \wedge (N'_{source}, t') \in F'\} \cup \{(t', p) \mid (t, p) \in F \wedge (t', N'_{sink}) \in F'\}$.
\end{itemize}
\end{definition}

Based on this operation, we define the class of \emph{separable WF-nets}. This class is defined by the property that its members can be recursively constructed using state machine and marked graph templates. 

\begin{definition}[Separable WF-net]\label{def:separable}
The class of \emph{separable WF-nets} is defined as follows:
\begin{enumerate}
    \item Every marked-graph WF-net is separable.
    \item Every state-machine WF-net is separable. 
    \item Let $N = (P, T, F)$ and $N' = (P', T', F')$ be two separable WF-nets such that $(P \cup T) \cap (P' \cup T') = \emptyset$. For any $t \in T$, $N_{[t \rightarrow N']}$ is separable.
\end{enumerate}
\end{definition}

The WF-net shown in \cref{fig:ex:wf} is separable. In a separable WF-net, concurrency and choice logic are hierarchically nested but never cross-linked. This structure allows for arbitrary complexity \emph{within} a decision logic (e.g., unstructured loops and jumps via choice graphs) and arbitrary complexity \emph{within} a concurrent region (e.g., unstructured synchronization via partial orders). However, it does not support the arbitrary interleaving of decision and concurrency logic (e.g., a choice that depends on the state of a parallel branch without synchronization). For example, the WF-net shown in \cref{fig:ex:not_sep_wf} is not separable.

\begin{figure}[!t]
    \centering    
        \includegraphics[width=0.6\textwidth]{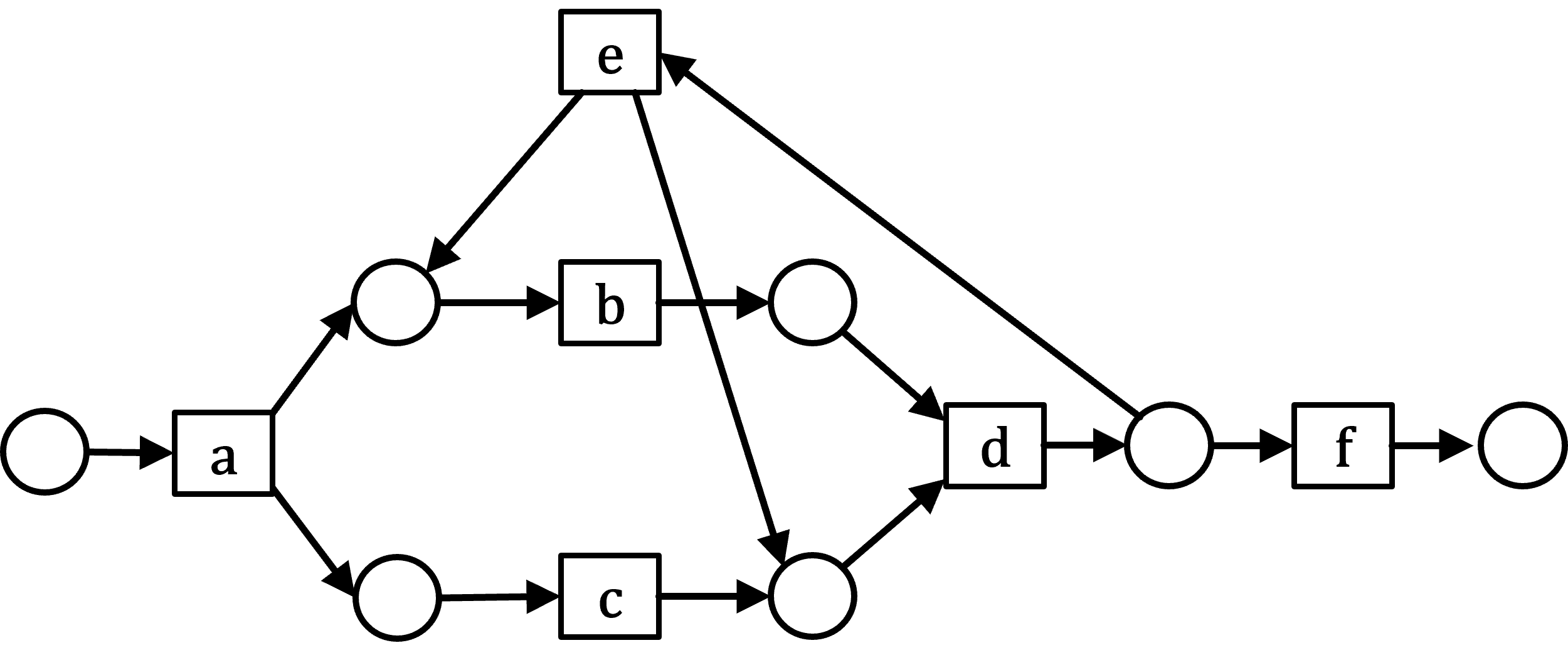}
        \caption{A free-choice WF-net that is not separable. \label{fig:ex:not_sep_wf}}
\end{figure}

The class of separable WF-nets represents a specific subset of well-structured workflow nets, distinguished by the following key properties:
\begin{itemize}
    \item \textbf{Subclass of Free-Choice WF-Nets:} State machines and marked graphs are both subclasses of free-choice Petri nets. The substitutive composition of free-choice SESE fragments preserves the free-choice property. Consequently, by construction, all separable WF-nets are inherently free-choice.
    
    \item \textbf{Absence of PT-Handles and TP-Handles:} The hierarchical structure of separable nets guarantees the absence of unstructured paths that create complex overlaps, often referred to as \emph{PT-handles} and \emph{TP-handles} \cite{DBLP:conf/apn/EsparzaS89}. A PT-handle occurs when two disjoint paths start at a common place and end at a common transition. Similarly, a TP-handle occurs when two disjoint paths start at a common transition and end at a common place. In separable nets, all splits and joins are strictly encapsulated within identified SESE blocks (corresponding to either state machines or marked graphs). This encapsulation ensures that no handles can cross the boundaries of these blocks in an unstructured manner. Therefore, all separable WF-nets are guaranteed to be free from PT-handles and TP-handles.
\end{itemize}

By construction, every POWL 2.0 model has an equivalent representation in the class of separable WF-nets. In \cref{sec:gua:redisc}, we show that our algorithm is complete for this class under the soundness and safeness requirement; i.e., we show the our algorithm successfully converts any sound and safe separable WF-net into a POWL 2.0 model that captures the same language.

\section{Transforming Workflow Nets into POWL}\label{sec:convert}
This section presents a generalized algorithm for transforming safe and sound WF-nets into equivalent POWL 2.0 models. Unlike our previous approach that relied on detecting specific block-structured patterns (XOR, Loop) \cite{DBLP:conf/apn/KouraniPA25}, our new approach leverages the unified expressiveness of POWL 2.0. The transformation is based on a recursive decomposition strategy that exploits the structural duality between the two main constructs of POWL 2.0: \emph{partial orders} and \emph{choice graphs}.

A POWL 2.0 model can be viewed as a hierarchy where each node organizes its children using either:
\begin{itemize}
    \item A partial order, which models concurrency and causal dependencies. To act as a valid parent node, a partial order structure must not contain external decision logic. Therefore, identifying this structure requires identifying and encapsulating all \emph{conflicts} (exclusive choices) into the child nodes.
    \item A choice graph, which models exclusive decisions and cycles. To act as a valid parent node, a choice graph structure must not contain external concurrency. Therefore, identifying this structure requires identifying and encapsulating all \emph{concurrency} (parallel splits and joins) into the child nodes.
\end{itemize}

Based on this observation, our algorithm recursively attempts to partition the WF-net into subnets by either ``hiding conflicts'' (to find a partial order) or ``hiding concurrency'' (to find a choice graph).

The remainder of this section is structured as follows. \cref{sec:preproc} outlines the optional preprocessing steps. \cref{sec:mine_po} defines the criteria and projection rules for identifying partial order structures. \cref{sec:mine_cg} defines the criteria and projection rules for identifying choice graph structures. Finally, \cref{sec:algo_main} presents the integrated recursive algorithm. 

\subsection{Preprocessing}\label{sec:preproc}
While our core conversion algorithm is designed for the class of separable WF-nets, many real-world models contain syntactic variations that obscure their underlying separable structure. To address this, we define an optional, but crucial \emph{preprocessing} phase that runs before the main conversion. This step applies a set of language-preserving reduction rules to transform a WF-net into a structurally equivalent, simpler form. The impact of this phase is significant, as it extends the practical scope of our approach.

Numerous reduction rules have been proposed in the literature, such as those presented in \cite{desel1995free,DBLP:conf/ac/Berthelot86,10.1007/BFb0016204,DBLP:journals/pieee/Murata89}. Any reduction rule can be applied as long as it preserves the essential structural properties of the WF-net, namely its language, safeness, and soundness. The specific reduction rules implemented in our approach are illustrated in \cref{fig:prepro}; they include the elimination of duplicate places and the introduction of explicit places to represent XOR-splits/joins. By applying these transformations, preprocessing can convert certain non-separable WF-nets into equivalent separable ones.

\begin{figure}[!t]
    \centering    
        \begin{subfigure}{0.48\textwidth}
            \centering
            \includegraphics[width=\textwidth]{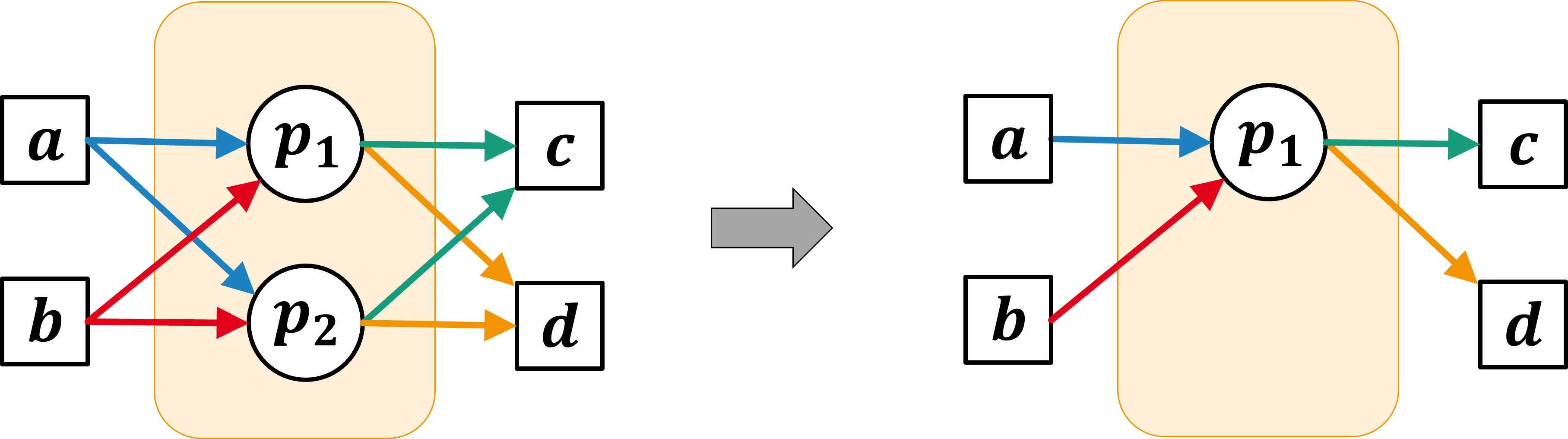}
            \caption{Filtering out duplicated places.}
        \end{subfigure}
        
        \begin{subfigure}{0.48\textwidth}
            \centering
            \includegraphics[width=\textwidth]{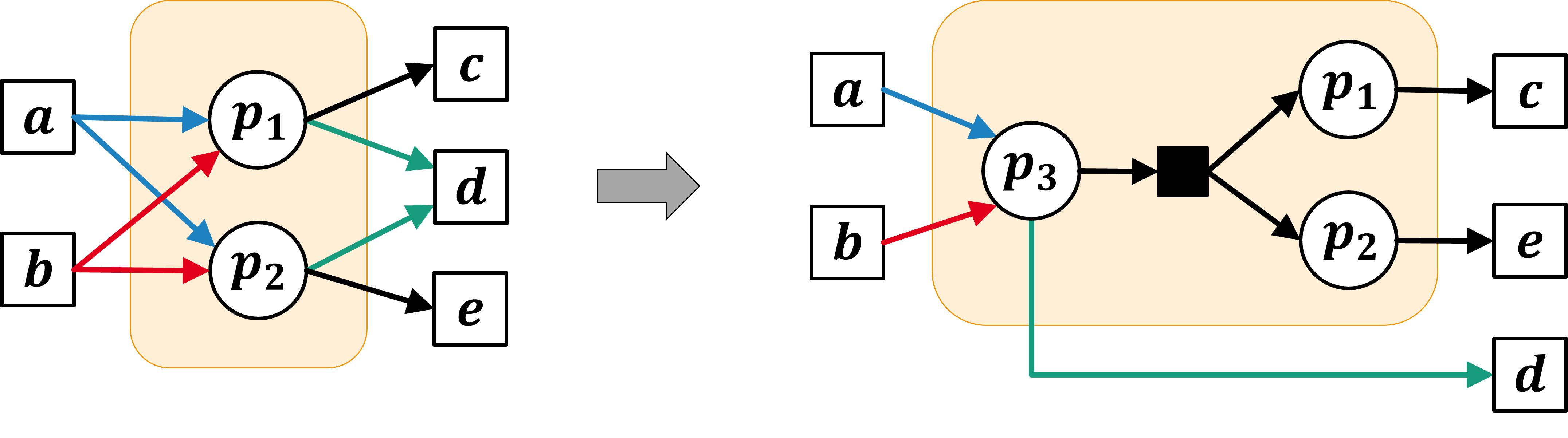}
            \caption{Introducing explicit XOR split places.}
        \end{subfigure}
        \begin{subfigure}{0.02\textwidth}
                $ $
        \end{subfigure}
        \begin{subfigure}{0.48\textwidth}
            \centering
            \includegraphics[width=\textwidth]{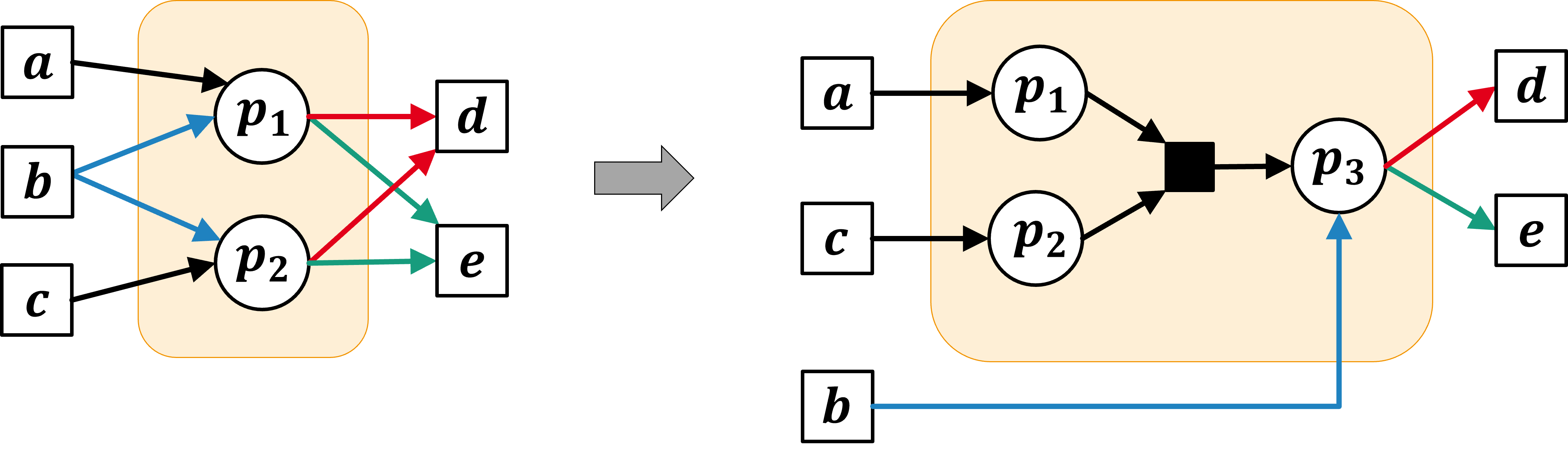}
            \caption{Introducing explicit XOR join places.}
        \end{subfigure}

        \caption{Reduction rules illustrated with examples.\label{fig:prepro}}
\end{figure}

\subsection{Mining Partial Orders}\label{sec:mine_po}
This section addresses partitioning a WF-net into smaller subnets such that the given WF-net corresponds to a partial order over the identified parts. 

To achieve this decomposition, our approach proceeds in three steps:
\begin{enumerate}
    \item \textbf{Detection:} We first identify a partition of the WF-net's transitions that hides all exclusive choices (\cref{sec:detect_po}).
    \item \textbf{Projection:} We then project the original net onto these identified parts to create valid sub-models for the next recursion step (\cref{sec:project_po}).
    \item \textbf{Composition:} Finally, we derive the execution order between these parts to define the top-level partial order structure (\cref{sec:order_po}).
\end{enumerate}

\subsubsection{Detecting Conflict-Hiding Partitions}\label{sec:detect_po}
Structurally, a partial order corresponds to a \emph{marked graph} pattern in Petri nets. In a marked graph, places serve solely as connectors that enforce causal dependencies; they do not introduce choices (XOR-splits with multiple outgoing arcs) nor do they merge alternative paths (XOR-joins with multiple incoming arcs). Therefore, to decompose a WF-net into a partial order, all decision logic must be fully encapsulated within the identified subnets.

To formalize these requirements, we define a \emph{conflict-hiding partition} as a grouping of transitions that enforces marked graph constraints on the interaction between distinct parts. Specifically, it ensures that no single place in the net acts as a bridge distributing tokens to multiple distinct parts or collecting tokens from multiple distinct parts. Furthermore, it ensures that the resulting parts represent SESE fragments, allowing them to be treated as atomic nodes within the hierarchy.

\paragraph*{Notation.}
Let $N=(P,T, F)$ be a WF-net. 
\begin{itemize}
    \item For $T' \subseteq T$, $\Pre{T'} = \{p \in P \ | \ T' \cap \post{p} \neq \emptyset \ \wedge \ (p = N_{source} \ \vee \ (T \setminus T') \cap \pre{p} \neq \emptyset)\}$ is the set of \emph{entry points} of $T'$.
    \item For $T' \subseteq T$, $\Post{T'} = \{p \in P \ | \ T' \cap \pre{p} \neq \emptyset \ \wedge \ (p = N_{sink} \ \vee \ (T \setminus T') \cap \post{p} \neq \emptyset)\}$ is the set of \emph{exit points} of $T'$.
    \item For $T' \subseteq T$, two places $p,p' \in P$ are \emph{equivalent with respect to $T'$}, denoted as $p\approx_{T'} p'$, iff $(\pre{p} \cap T' = \pre{p'} \cap T') \ \wedge \ (\post{p} \cap T' = \post{p'} \cap T')$.
\end{itemize}

\begin{definition}[Conflict-Hiding Partition]\label{def:po_pattern}
Let $N = (P, T, F)$ be a safe and sound WF-net. A partition $G = \{T_1, \dots, T_n\}$ of $T$ is \emph{conflict-hiding } iff the following conditions hold: 
\begin{enumerate}
    \item \textbf{No Top-Level XOR-Splits:} for every place $p \in P$:
    \[ |\{T_i \in G \mid p \in \Pre{T_i}\}| \leq 1. \]
    \item \textbf{No Top-Level XOR-Joins:} for every place $p \in P$:
    \[ |\{T_i \in G \mid p \in \Post{T_i}\}| \leq 1. \]
    \item \textbf{Single Entry Fragments:} for all $i \in \{1, \ldots, n\}$ and $p, p' \in \Pre{T_i}$: $p \approx_{T_i} p'$.
    \item \textbf{Single Exit Fragments:} for all $i \in \{1, \ldots, n\}$ and $p, p' \in \Post{T_i}$: $p \approx_{T_i} p'$.

\end{enumerate}
\end{definition} 

Conditions 1 and 2 enforce the marked graph structural duality at the top level. Conditions 3 and 4 require that all entry places ($\Pre{T_i}$) and exit places ($\Post{T_i}$) are equivalent with respect to $T_i$ (not leading to different internal behaviors within $T_i$). This ensures that the identified parts form cleanly separable components that can be executed independently with well-defined entry and exit points. \cref{fig:order} shows a conflict-hiding partition detected for an example WF-net.

\begin{figure}[!t]
    \centering    
     \includegraphics[width=\textwidth]{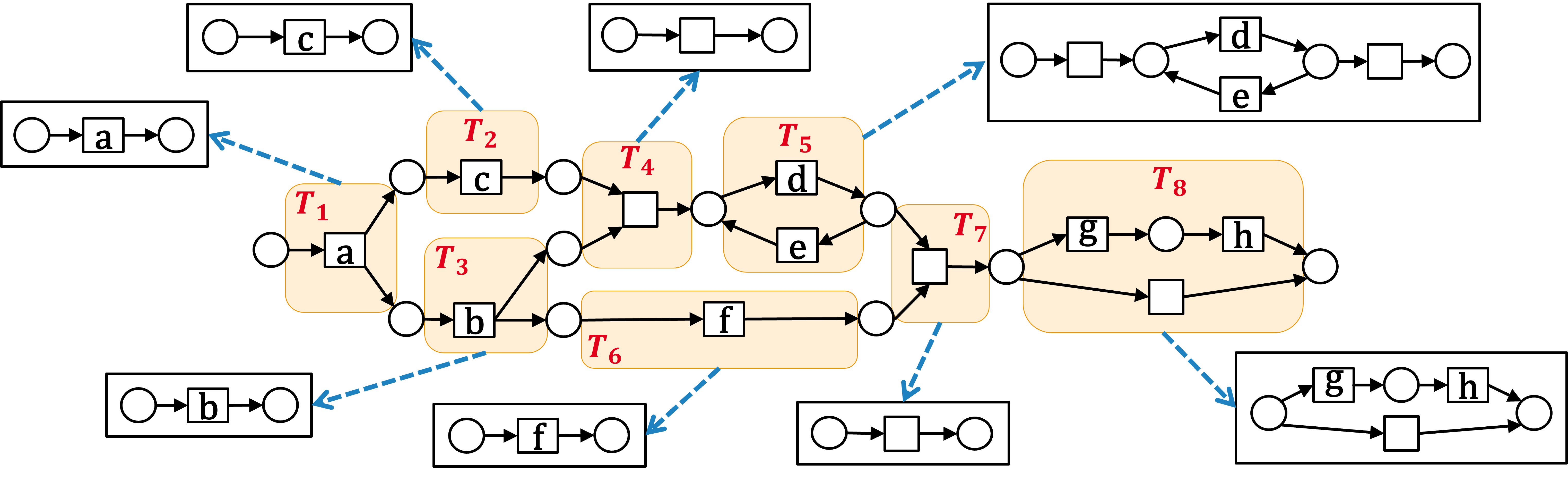}
    \caption{An example illustrating a conflict-hiding partition $\{T_1, \dots, T_8\}$ and the partial order projection of the WF-net on the identified parts.\label{fig:order}}
\end{figure}

To generate a partition that satisfies the requirements of \cref{def:po_pattern}, we employ a symmetric algorithmic strategy as outlined in \cref{alg:partial_order_detect}. The core reasoning is centered on identifying the ``conflict scope'' of every split and merge point. If a place $p$ has transitions reachable from one outgoing flow but not all others, then these transitions are considered \textit{exclusive} to the decision branch introduced by the place and must all be merged into the same component (lines 3-8). Similarly, if a place $p$ has transitions from which only specific incoming flows are reachable, these transitions are exclusive to that alternative merging path and must also be merged (lines 9-14). By applying this logic bidirectionally, the algorithm ensures that all decision logic is pulled together from both ends and aggregated into a single atomic part, thereby preserving the structural properties of a marked graph at the top level of the partition.




\begin{algorithm}[!t]
\caption{Conflict-Hiding Partitioning in a WF-Net ($\popart$).}\label{alg:partial_order_detect}
\DontPrintSemicolon
\KwIn{A safe and sound WF-net $N = (P, T, F)$.}
\KwOut{A partition  $G = \{T_1, \dots, T_n\}$ representing marked graph nodes.}
\SetKwFunction{FMain}{$\popart$}
\SetKwProg{Fn}{Function}{:}{}
\Fn{\FMain{$N$}}{
    $G \leftarrow \{\{t\} \ | \ t \in T\}$ 

    \tcp{Forward Analysis: XOR-Splits}
    \For{$p_{split} \in P$ such that $|\post{p_{split}}| > 1$}{
        $group \leftarrow {\{t \in T \ | \ \exists_{t_1, t_2 \in \post{_{split}}} \quad t_1 \TTR t \ \wedge \ t_2 \notTTR t \}}$\;

        \If{$\card{group} > 1$}{
            $G \leftarrow G \setminus \{G_t \mid t \in group\} \cup \big\{ \bigcup\limits_{t \in group} G_{t} \big\}$\;
        }
    }

    \tcp{Backward Analysis: XOR-Joins}
    \For{$p_{join} \in P$ such that $|\pre{p_{join}}| > 1$}{

        $group \leftarrow {\{t \in T \ | \ \exists_{t_1, t_2 \in \pre{p_{join}}} \quad t \TTR t_1 \wedge t \notTTR t_2 \}}$\;

        \If{$\card{group} > 1$}{
            $G \leftarrow G \setminus \{G_t \mid t \in group\} \cup \big\{ \bigcup\limits_{t \in group} G_{t} \big\}$\;
        }
    }

    \KwRet{$G$}
}
\end{algorithm}

The partition $\popart(N)$ returned by \cref{alg:partial_order_detect} is a conflict-hiding partition if it consists of at least two parts and the conditions of \cref{def:po_pattern} are met. 

\subsubsection{Projecting Child Models}\label{sec:project_po}
Once a valid partition is identified, the next step is to isolate the specific subnet corresponding to each part for recursive processing.

\paragraph{Normalization.} Before defining the projection for conflict-hiding partitions, we introduce the concept of \emph{normalization}. Normalization aims at ensuring conformance of the projected net with the requirements of WF-nets (cf. \cref{def:wf-net}) by adding new source and sink places if needed. Let $N$ be a Petri net with known unique start and end places $p_s \in P$ and $p_e \in P$, respectively. Then $P$ is normalized into a new Petri net $\mathit{Normalize}(N, p_s, p_e)$ by (i) inserting a new start place and connecting it to $p_s$ through a silent transition in case $\pre{p_s} \neq \emptyset$ and (ii) inserting a new end place and connecting it to $p_e$ through a silent transition in case $\post{p_e} \neq \emptyset$. 

After detecting a conflict-hiding partition, the WF-net is projected on the identified parts as illustrated in the example shown in \cref{fig:order}. This projection is done by selecting the appropriate subset of transitions, adding unique start and end places to represent the entry and exit points of the part, adjusting the flow relation accordingly, and applying normalization if needed.

\begin{definition}[Partial Order Projection]\label{def:projection_po}
Let $N = (P, T, F)$ be a safe and sound WF-net. Let $G$ be a conflict-hiding partition and $T' \in G$ be a part. The partial order projection of $N$ on $T'$ is $\poproject(N, T') = \mathit{Normalize}(N', p_s, p_e)$ where $p_s, p_e \notin P$ are two fresh places and $N' = (P', T', F')$ is constructed as follows:
\begin{itemize}
    \item $P' = (P\project{T'} \setminus (\Pre{T'} \cup \Post{T'})) \cup \{p_s, p_e\}$.
    \item $F' = F\project{P', T'}$\\ 
    $\cup \ \{(p_s, t) \ | \ \exists p \in \Pre{T'}: (p, t) \in F \} \ \cup \ \{(t, p_s) \ | \ \exists p \in \Pre{T'}: (t, p) \in F \}$\\
    $\cup \ \{(p_e, t) \ | \ \exists p \in \Post{T'}: (p, t) \in F \} \ \cup \ \{(t, p_e) \ | \ \exists p \in \Post{T'}: (t, p) \in F \}$.
\end{itemize}
\end{definition}

\subsubsection{Generating the Top-Level Partial Order}\label{sec:order_po}

After projecting the input WF-net $N$ onto each part $T_i$ and recursively converting it into a POWL submodel $\powl_i$, these child models must be combined into a single structure that mirrors the top-level behavior of the original net. 

Since the conflict-hiding partition ensures that the high-level control flow behaves like a marked graph, the dependencies between child models are derived by the flow of tokens through the places connecting them. If a place acts as an exit for $T_i$ and an entry for $T_j$, it establishes a direct causal link. Because the partition satisfies the unique local interface requirement, this dependency is uniform: any place in the intersection $(\Post{T_i} \cap \Pre{T_j})$ effectively acts as a synchronization barrier between the two sub-processes. By extracting these links and calculating their transitive closure, we derive the partial order operator that assembles the children into the final POWL model.

\begin{definition}[Execution Order]\label{def:execution_order}
Let $G = \{T_1, \dots, T_n\}$ be a conflict-hiding partition of a safe and sound WF-net $N$. The \emph{execution order} of $G$ within $N$ is the binary relation $\po = \mathit{order}(N, G)$ over the indices $\{1, \dots, n\}$ defined as follows:
\[
i \po j \Leftrightarrow (\Post{T_i} \cap \Pre{T_j}) \neq \emptyset.
\]
\end{definition}

For a valid conflict-hiding partition in a sound WF-net, the transitive closure of $\po = \mathit{order}(N, G)$ is guaranteed to be a strict partial order. This guarantee holds due to the absence of decision points in the top-level marked graph structure. The relationship between any two parts $T_i$ and $T_j$ is purely one of synchronization: if $i \po j$, part $T_j$ \emph{must wait} for part $T_i$ to complete before it can start. If the execution order contained a cycle of any length (e.g., $T_i \to \dots \to T_i$), it would imply a circular dependency where $T_i$ is effectively waiting for itself to finish before it can start. In a workflow net, such a circular wait constitutes a deadlock, violating the soundness assumption.

\subsection{Mining Choice Graphs}\label{sec:mine_cg}
This section addresses partitioning a WF-net into subnets such that the given WF-net corresponds to a choice graph over the identified parts.

To achieve this decomposition, our approach follows three steps:
\begin{enumerate}
    \item \textbf{Detection:} We first identify a partition that hides all parallel splits and joins (\cref{sec:detect_cg}).
    \item \textbf{Projection:} We then project the original net onto these parts to create the sub-models (\cref{sec:project_cg}).
    \item \textbf{Composition:} Finally, we derive the execution flow between these parts to define the top-level choice graph (\cref{sec:flow_cg}).
\end{enumerate}

\subsubsection{Detecting Concurrency-Hiding Partitions}\label{sec:detect_cg}
To decompose a WF-net into a choice graph, the high-level structure must behave like a \emph{state machine} (i.e., tokens follow exactly one path, and there is no active concurrency at the top level). In standard Petri net semantics, concurrency is introduced by transitions with multiple output arcs (AND-splits) and synchronized by transitions with multiple input arcs (AND-joins). Therefore, identifying a state machine structure requires identifying and encapsulating these concurrency patterns within the subnets. This leads to the definition of a \emph{concurrency-hiding partition}, which requires each part to interact with the rest of the net through exactly one entry place and one exit place.

\begin{definition}[Concurrency-Hiding Partition]\label{def:cg_pattern}
Let $N = (P, T, F)$ be a safe and sound WF-net. A partition $G = \{T_1, \dots, T_n\}$ of $T$ is a \emph{concurrency-hiding} iff for all $i \in \{1, \ldots, n\}$:
\[
    |\Pre{T_i}| = 1 \quad \wedge \quad |\Post{T_i}| = 1.
\]
\end{definition}

This simple definition effectively forces the encapsulation of all concurrent threads as it requires each abstracted part to satisfy the unique entry and exit requirements of a state machine. \cref{fig:cg_proj} shows a concurrency-hiding partition for an example WF-net.

\begin{figure}[!t]
    \centering    
     \includegraphics[width=\textwidth]{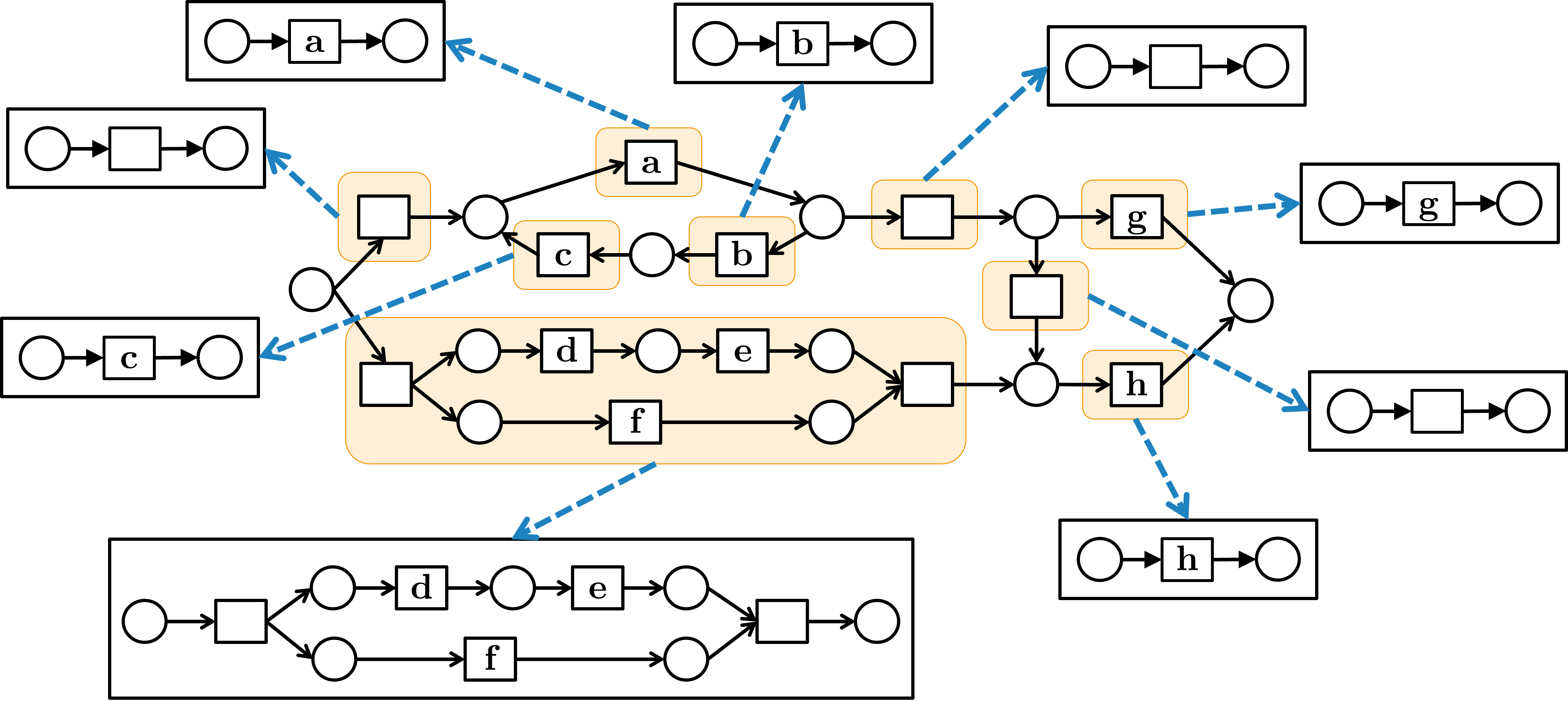}
    \caption{An example illustrating a concurrency-hiding partition and the choice graph projection of the WF-net on the identified parts.\label{fig:cg_proj}}
\end{figure}

To generate such a partition, we employ an algorithmic approach that explicitly targets and merges concurrency patterns. We first introduce the concepts of \textit{forward restricted reachability} and \textit{backward restricted reachability}. These concepts allow us to identify the set of transitions reachable from a specific node without passing through a designated ``stop'' node. This helps in isolating the specific branches of a parallel split or the specific threads feeding into a parallel join.

\begin{definition}[Forward Restricted Reachability]
Let $N = (P, T, F)$ be a Petri net, $p \in P$ be a place, and $t_{stop} \in T$ be a transition. The set of transitions reachable from $p$ via a path that does not visit $t_{stop}$, denoted as $\ReachAvoid{p}{t_{stop}}$, is defined as follows:
\[
t \in \ReachAvoid{p}{t_{stop}} \iff \exists\ t_1, \dots, t_{n} \in T \text{ and } p_1, \dots, p_{n} \in P \]
such that
\[ 
p_1 = p \ \wedge \ t_{n} = t,
\]
and for each $i$ ($1 \leq i \leq n$):
\[
t_i \neq t_{stop} \ \wedge \ (p_i, t_i) \in F,
\]
and for each $i$ ($1 \leq i < n$):
\[
(t_i, p_{i+1}) \in F.
\]

\end{definition}

\begin{definition}[Backward Restricted Reachability]
Let $N = (P, T, F)$ be a Petri net, $p \in P$ be a place, and $t_{stop} \in T$ be a transition. The set of transitions from which $p$ is reachable via a path that does not visit $t_{stop}$, denoted as $\BackReachAvoid{p}{t_{stop}}$, is defined as follows:
\[
t \in \BackReachAvoid{p}{t_{stop}} \iff \exists\ t_1, \dots, t_{n} \in T \text{ and } p_1, \dots, p_{n} \in P \]
such that
\[
t_1 = t \ \wedge \ p_{n} = p,
\]
and for each $i$ ($1 \leq i \leq n$):
\[
t_i \neq t_{stop} \ \wedge \ (t_i, p_i) \in F,
\]
and for each $i$ ($1 \leq i < n$):
\[
(p_i, t_{i+1}) \in F.
\]
\end{definition}

Using these concepts, we define an algorithm that iteratively merges transitions to hide concurrency. \cref{alg:choice_graph_detect} outlines the strategy, which operates in two symmetric phases. In the first step (lines 3-9), for every AND-split transition $t_{split}$, the algorithm identifies transitions that are reachable from a specific outgoing branch $p_1 \in \post{t_{split}}$ but not all branches (using $\ReachAvoid{}{t_{split}}$). These are merged with $t_{split}$, encapsulating the source of the concurrency and the parallel ``threads'' it spawns. In the second step (lines 10-16), the algorithm targets AND-joins. For a join transition $t_{join}$, it identifies transitions that can reach a specific incoming branch $p_1 \in \pre{t_{join}}$ but not all branches (using $\BackReachAvoid{}{t_{join}}$). These are merged with $t_{join}$, capturing the synchronization of the concurrency and parallel threads feeding into it. This bidirectional approach ensures that even in complex nets, the internal nodes of a concurrent region are pulled together from both ends.

\begin{algorithm}[!t]
\caption{Concurrency-Hiding Partitioning in a WF-Net ($\cgpart$).}\label{alg:choice_graph_detect}
\DontPrintSemicolon
\KwIn{A safe and sound WF-net $N = (P, T, F)$.}
\KwOut{A partition $G = \{T_1, \dots, T_n\}$ representing state machine nodes.}
\SetKwFunction{FMain}{$\cgpart$}
\SetKwProg{Fn}{Function}{:}{}

\Fn{\FMain{$N$}}{
    $G \leftarrow \{\{t\} \ | \ t \in T\}$\; 

    \tcp{Forward Analysis: AND-Splits}
    \For{$t_{split} \in T$ such that $|\post{t_{split}}| > 1$}{
        
        $threads \leftarrow \{t \in T \ | \ \exists_{p_1, p_2 \in \post{t_{split}}} \ t \in \ReachAvoid{p_1}{t_{split}} \ \wedge \ t \notin \ReachAvoid{p_2}{t_{split}} \}$\;

        $group \leftarrow \{t_{split}\} \cup threads$\;

        \If{$\card{group} > 1$}{
    
            $G \leftarrow G \setminus \{G_t \mid t \in group\} \cup \big\{ \bigcup\limits_{t \in group} G_{t} \big\}$\;

        }
        
    }

    \tcp{Backward Analysis: AND-Joins}
    \For{$t_{join} \in T$ such that $|\pre{t_{join}}| > 1$}{
 
        $threads \leftarrow \{t \in T \ | \ \exists_{p_1, p_2 \in \pre{t_{join}}} \ t \in \BackReachAvoid{p_1}{t_{join}} \ \wedge \ t \notin \BackReachAvoid{p_2}{t_{join}} \}$\;

        $group \leftarrow \{t_{join}\} \cup threads$\;

        \If{$\card{group} > 1$}{

            $G \leftarrow G \setminus \{G_t \mid t \in group\} \cup \big\{ \bigcup\limits_{t \in group} G_{t} \big\}$\;

        }

        
    }
    
    \KwRet{$G$}
}
\end{algorithm}

The partition $\cgpart(N)$ returned by \cref{alg:choice_graph_detect} is a valid concurrency-hiding partition if it consists of at least two parts and satisfies the requirement of \cref{def:cg_pattern} (i.e., $|\Pre{T_i}| = 1$ and $|\Post{T_i}| = 1$ for all $T_i \in G$).



\subsubsection{Projecting Child Models}\label{sec:project_cg}
Once the concurrency-hiding partition is identified, we isolate the subnet for each part. The projection is similar to the partial order case, as illustrated in the example shown in \cref{fig:cg_proj}. However, since $|\Pre{T'}| = 1 = |\Post{T'}|$ in a concurrency-hiding partition, the construction is simplified; no unification of multiple start/end places is required.

\begin{definition}[Choice Graph Projection]\label{def:projection_cg}
Let $N = (P, T, F)$ be a safe and sound WF-net. Let $G$ be a concurrency-hiding partition and $T' \in G$ be a part. The choice graph projection of $N$ on $T'$ is $\cgproject(N, T') = \mathit{Normalize}(N', p_s, p_e)$ where:
\begin{itemize}
    \item $P' = P\project{T'}$,
    \item $F' = F\project{P', T'}$,
    \item $N' = (P', T', F')$,
    \item $p_s \in P'$ is the unique entry place in $\Pre{T'}$,
    \item $p_e \in P'$ is the unique exit place in $\Post{T'}$.
    
\end{itemize}
\end{definition}

\subsubsection{Generating the Top-Level Choice Graph}\label{sec:flow_cg}

After the input WF-net $N$ is projected onto each part $T_i$ and recursively converted into a POWL submodel $\powl_i$, these child models must be assembled into a single hierarchical node that captures the overall decision logic of the original net. At the top-level state machine structure, there is no active concurrency; instead, the process moves from one part to another by following alternative paths. To represent this logic, we construct a choice graph where the nodes correspond to the submodels and the edges represent the possible flow of a single token through the process. 

\begin{definition}[Execution Flow]\label{def:flow_graph}
Let $G = \{T_1, \dots, T_n\}$ be a concurrency-hiding partition of a safe and sound WF-net $N$. The \emph{execution flow} of $G$ within $N$ is the choice graph $\mathit{flow}(N, G) = (V, E) \in \Graphs{n}$ with 
\[
    V = \{1, \dots, n\} \cup \{\xorstart{}, \xorend{}\}
\]
and 
\[
    E = \{(i, j) \mid (\Post{T_i} \cap \Pre{T_j}) \neq \emptyset\} \ \cup \ \{(\xorstart{}, i) \mid N_{source} \in \Pre{T_i}\} \ \cup \ \{(i, \xorend{}) \mid N_{sink} \in \Post{T_i}\}.
\]

\end{definition}

While $\mathit{order}(N, G)$ captures the internal causal dependencies between sub-components, $\mathit{flow}(N, G)$ encapsulates the entire routing logic of the choice graph, including its entry from the workflow source and its exit to the workflow sink. This allows us to treat the entire structure as an abstract choice graph template that can be instantiated with concrete submodels during the recursive conversion process.

\subsection{WF-Net to POWL 2.0 Converter}\label{sec:algo_main}

\begin{algorithm}[!t]
\caption{Conversion of a WF-Net into POWL 2.0.}\label{alg:convert_pn_to_powl_v2}
\DontPrintSemicolon
\KwIn{A safe and sound WF-net $N = (P, T, F)$.}
\KwOut{A POWL 2.0 model or a fallback representation.}

\SetKwFunction{FMain}{ConvertNetToPOWL}
\SetKwFunction{FallThrough}{FallThrough}
\SetKwProg{Fn}{Function}{:}{}
\Fn{\FMain{$N$}}{
    
    \tcp{(1) Base case}
        \If{$\card{T} = 1$ with $T = \{t\}$, $\card{P} = 2$, and $F = \{(N_{source}, t), (t, N_{sink})\}$}{
            \Return $t$\;
        }

    \tcp{(2) Attempt to decompose into a marked graph}   
       $G = \{T_1, \dots, T_n\} \leftarrow \popart(N)$\;
       \If{$\card{G} > 1$ and $G$ is conflict-hiding}{
            \tcp{Ensure structural progress}
            \If{$\nexists_{T_i \in G} \poproject(N, T_i) \cong N$}{ 
                \For{$T_i \in G$}{
                    $\powl_i \gets$ \FMain{$\poproject(N, T_i)$}\;
                }
                $\po \gets \closure{\mathit{order}}(N, G)$\;
                \Return $\po(\powl_1, \dots, \powl_n)$\; 
            }
        }

    \tcp{(3) Attempt to decompose into a state machine}
    $G = \{T_1, \dots, T_n\} \leftarrow \cgpart(N)$\;
    \If{$\card{G} > 1$ and $G$ is concurrency-hiding}{
        \tcp{Ensure structural progress}
        \If{$\nexists_{T_i \in G} \cgproject(N, T_i) \cong N$}{ 
            \For{$T_i \in G$}{
                $\powl_i \gets$ \FMain{$\cgproject(N, T_i)$}\;
            }
            $\cg \gets \mathit{flow}(N, G)$\;
            \Return $\cg(\powl_1, \dots, \powl_n)$\;
        }
    }

    \Return \FallThrough{$N$}\;
}
\end{algorithm}

\cref{alg:convert_pn_to_powl_v2} converts a safe and sound WF-net into an equivalent POWL 2.0 model. 
First, the algorithm checks whether the WF-net consists of a single transition (base case). If a base case is not detected, the algorithm attempts to identify a conflict-hiding partition (to form a partial order) or a concurrency-hiding partition (to form a choice graph). If a partition is found, the algorithm projects the WF-net on the identified parts, recursively converts the created subnets into POWL models, and combines them using the appropriate POWL structure. 

Despite the high expressiveness of POWL 2.0, it remains a formal subclass of sound WF-nets. If neither a partition nor a base case is detected, \cref{alg:convert_pn_to_powl_v2} reverts to a conceptual \emph{fall-through} mechanism. While the specific implementation of this mechanism is outside the scope of this paper, several strategies could be employed to ensure a result is always produced: 
\begin{itemize}
    \item \textbf{Approximation/Generalization}: The algorithm could attempt to find a POWL 2.0 model that approximates the language of the irreducible sub-net, accepting a loss of precision. This could involve techniques to automatically mine the closest POWL 2.0 representation, or, as a last resort in cases of highly unstructured behavior, replacing the fragment with a \emph{flower model} (a POWL model that allows for any behavior over the involved activities).    
    \item \textbf{Transition Duplication}: Advanced techniques may be employed to transform certain non-separable fragments into equivalent separable structures by duplicating transitions.
    \item \textbf{Hybrid Modeling}: The irreducible fragment may be preserved as a nested WF-net leaf node within the parent POWL hierarchy, allowing for a mix of graph-based and hierarchical modeling.
\end{itemize}

In our current implementation and the subsequent evaluation (\cref{sec:eval}), we treat the \emph{fall-through} as a conversion failure. Specifically, if the algorithm invokes the fall-through at any level of recursion, the transformation is marked as unsuccessful, indicating that the original WF-net logic cannot be natively represented in POWL 2.0. This approach allows us to strictly and transparently quantify the practical coverage of the POWL 2.0 language on real-world business logic, ensuring that our results reflect the expressive power of the language itself.

Note that the order in which the algorithm searches for the different partitions is semantically irrelevant. The only scenario where a WF-net may simultaneously satisfy the requirements for both a conflict-hiding partition and a concurrency-hiding partition is when the top-level process exhibits strict sequential behavior. In such cases, the process satisfies the structural definitions of both a marked graph and a state machine. It can be modeled as either a total order ($\po \in \Orders{n}$ where $i \po j \Leftrightarrow i < j$) or a linear choice graph ($\cg \in \Graphs{n}$ where edges form a path $\xorstart{N} \to 1 \to \dots \to n \to \xorend{N}$), which are semantically equivalent.

\subsection*{Example Applications} 
\cref{fig:ex:powl} shows the POWL model generated by applying \cref{alg:convert_pn_to_powl_v2} on the WF-net from \cref{fig:ex:wf}. This net is free-choice and sound, yet its decision points and cycles are interleaved in a way that does not conform to standard XOR or loop blocks. Our previous conversion algorithm from \cite{DBLP:conf/apn/KouraniPA25} would fail on this model.

To illustrate the importance of preprocessing, we consider a case where the initial algorithm fails but can be rectified via preprocessing. \cref{fig:pre_success} shows a WF-net where places $p_2$ and $p_3$ model a choice between $d$ or executing both $b$ and $c$ concurrently. When attempting to identify a partial order pattern directly, the requirements are violated (e.g., no single entry in the sub-part $\{b,c,d\}$ since $\post{p_2} \cap \{b,c,d\} = \{b, d\} \neq \{c, d\} = \post{p_3} \cap \{b,c,d\}$). However, applying the reduction rules from \cref{fig:prepro} before applying \cref{alg:convert_pn_to_powl_v2} transforms the net into a separable structure, enabling successful conversion into POWL.


\begin{figure}[!t]
    \centering    
    \begin{subfigure}{\textwidth}
        \centering
        \includegraphics[width=\textwidth]{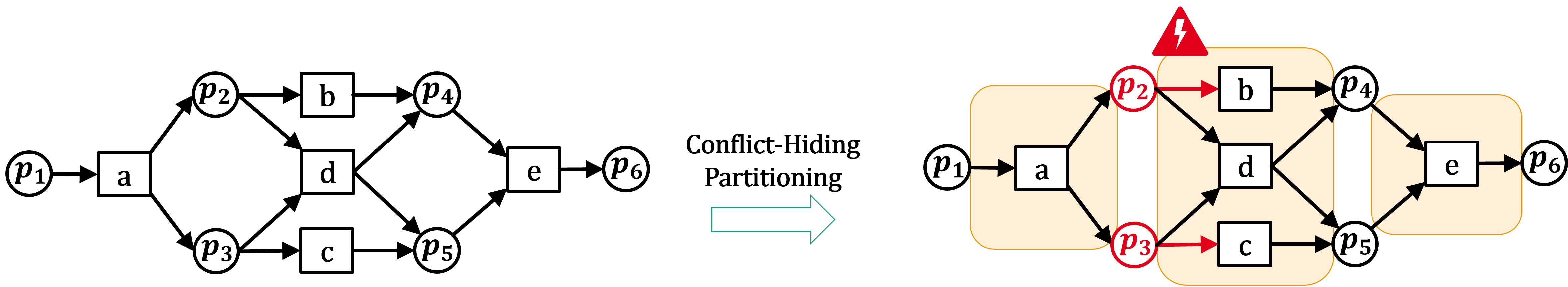}
        \caption{A WF-net where decision points combine choice with concurrency in a way that hinders pattern detection.}\label{fig:pre_success:fail}
    \end{subfigure}

    \begin{subfigure}{\textwidth}
        \centering
         \includegraphics[width=\textwidth]{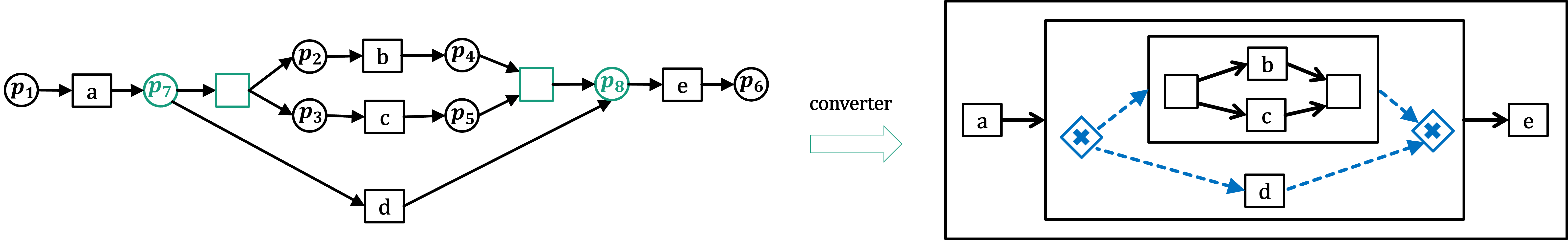}
        \caption{Result of applying the reduction rules (cf. \cref{fig:prepro}) to the net above, followed by a successful conversion.}\label{fig:pre_success:work}
    \end{subfigure}
    \caption{Impact of preprocessing on conversion success.\label{fig:pre_success}}
\end{figure}

\begin{figure}[!t]
\centering
   
    \begin{subfigure}{\textwidth}
    \centering
    \includegraphics[width=0.8\textwidth]{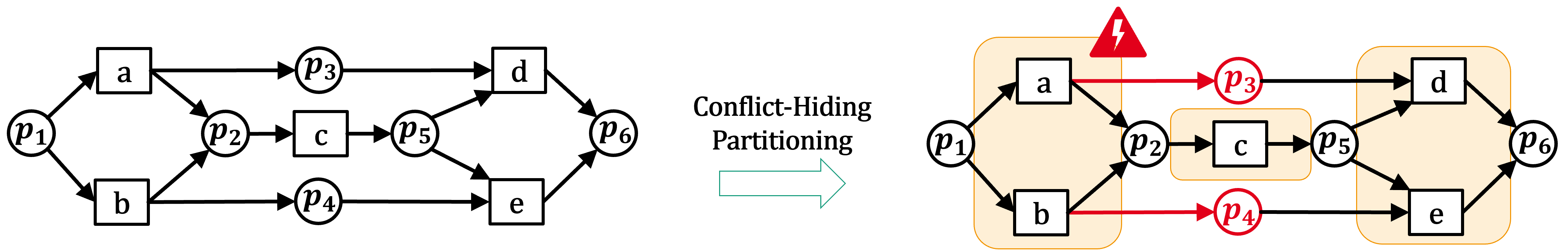}
    \caption{Non-free-choice WF-net with a long-term dependency between choices.}\label{fig:neg_ex:2}
    \end{subfigure}

    \begin{subfigure}{\textwidth}
    \centering
    \includegraphics[width=\textwidth]{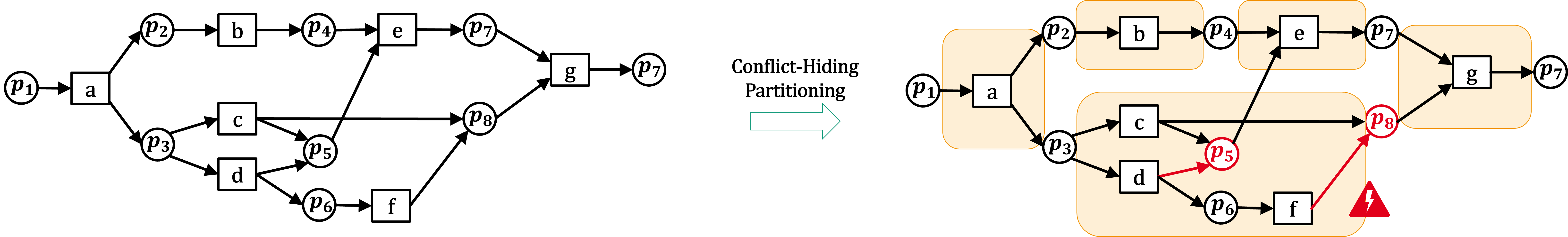}
    \caption{Free-choice WF-net with a long-term dependency. This model follows the structural pattern from \cite[Figure 25]{DBLP:journals/cj/PolyvyanyyGFW14}}\label{fig:neg_ex:4}
    \end{subfigure}

    \begin{subfigure}{\textwidth}
    \centering
    \includegraphics[width=\textwidth]{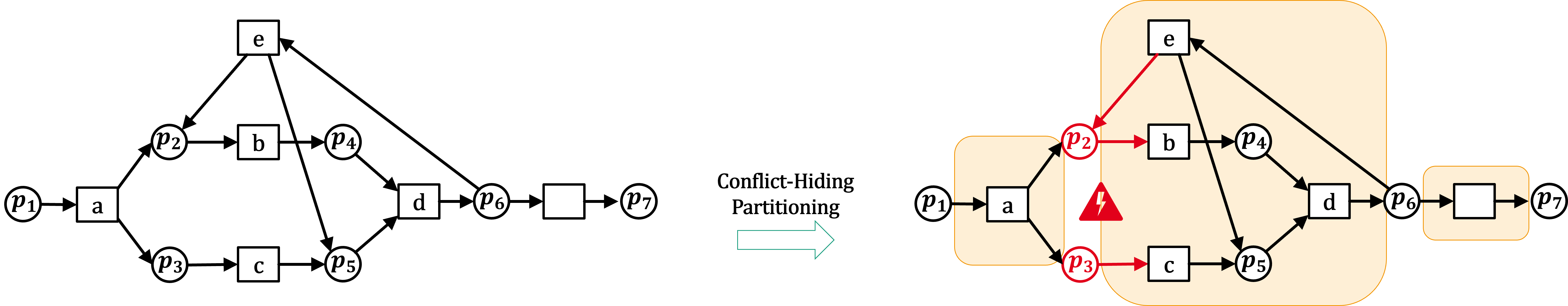}
    \caption{Free-choice WF-net where a loop jumps into a parallel branch.}\label{fig:neg_ex:5}
    \end{subfigure}

    \begin{subfigure}{\textwidth}
    \centering
    \includegraphics[width=\textwidth]{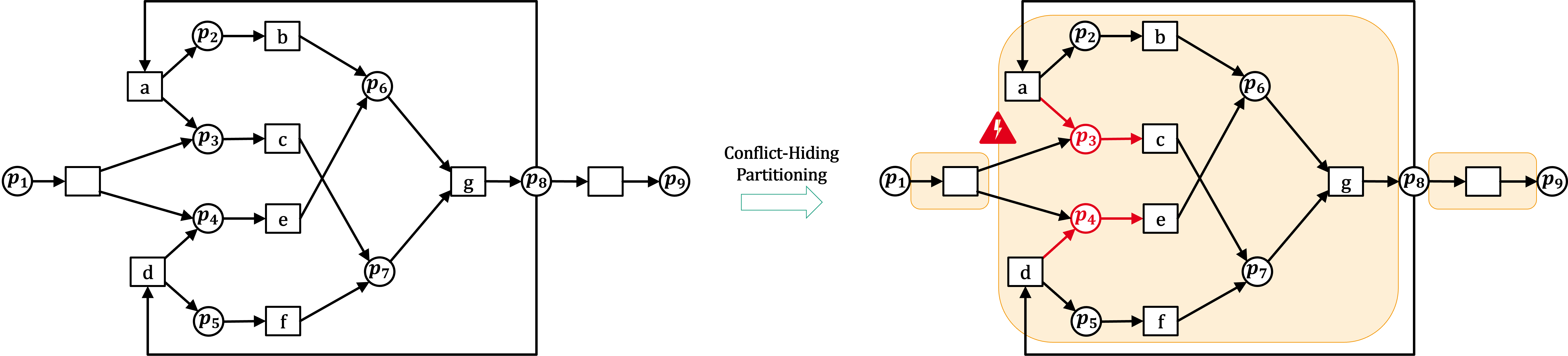}
    \caption{Free-choice WF-net with interleaved choice and concurrency sharing a synchronization point.}\label{fig:neg_ex:6}
    \end{subfigure}

    \caption{Examples where the decomposition attempts in \cref{alg:convert_pn_to_powl_v2} fail.\label{fig:neg_ex}}
\end{figure}

Despite these advancements, certain WF-nets remain outside the expressive scope of POWL 2.0. \cref{fig:neg_ex} illustrates example scenarios where the decomposition attempts in \cref{alg:convert_pn_to_powl_v2} fail, reverting to the fall-through:

\begin{itemize}
    \item The first WF-net (\cref{fig:neg_ex:2}) exhibits a choice between activities $a$ and $b$, followed by a non-free-choice between $d$ and $e$ that is influenced by the preceding choice. This long-term dependency choice cannot be represented in POWL. \cref{alg:convert_pn_to_powl_v2} first attempts to identify a conflict-hiding partition, resulting in a partition that violates the single exit requirement of \cref{def:po_pattern} for the first part $T_1 = \{a, b\}$ ($p_3 \not\approx_{T_1} p_4$). \cref{alg:convert_pn_to_powl_v2} also fails to detect a concurrency-hiding partition in the second step, as it merges all transitions into the same part.
    
     \item The second WF-net (\cref{fig:neg_ex:4}) illustrates a sound free-choice WF-net that is nonetheless non-separable. In this process, the lower thread chooses between $c$ and $d$. Both options synchronize with the upper thread at $e$ (via place $p_5$). However, if $d$ is chosen, an additional activity $f$ is executed concurrently with $e$. When \cref{alg:convert_pn_to_powl_v2} attempts to construct a partial order, it groups the decision logic $\{c, d, f\}$ into a single part. However, as highlighted in red, this partition has two distinct exit interfaces: $p_5$ (which leads to $e$) and $p_8$ (which leads to $g$). These exit places are not equivalent with respect to the part ($p_5 \not\approx p_8$), violating the single exit requirement (Condition 4 of \cref{def:po_pattern}). When the algorithm attempts to construct a choice graph, it ends up merging all transitions into a single part. Therefore, neither a partial order nor a choice graph can be detected.
    
    \item The third WF-net (\cref{fig:neg_ex:5}) depicts a sound free-choice process with a complex loop structure. The process starts with a parallel split at $a$, initiating $b$ and $c$. These branches synchronize at $d$. A decision is then made at $p_6$: either finish via a silent transition or loop back via $e$. Crucially, the loop transition $e$ places tokens in $p_2$ and $p_5$. This restarts activity $b$ but effectively skips $c$ (as the token for $c$'s completion is provided directly to $p_5$). When \cref{alg:convert_pn_to_powl_v2} attempts to isolate the loop logic into a partial order partition, it groups $\{b, c, d, e\}$. However, this part violates the single entry requirement (Condition 3 of \cref{def:po_pattern}). The entry places for this part are $p_2$ and $p_3$ (fed by the external transition $a$). These places are not equivalent with respect to the part: $p_2$ triggers $b$, while $p_3$ triggers $c$. The algorithm then attempts to construct a choice graph, merging all core transitions into a single part and only isolating the final silent transition. Because this results in a non-progressive projection isomorphic to the net itself ($N_i \cong N$), the choice graph attempt also fails.

    \item The fourth WF-net (\cref{fig:neg_ex:6}) demonstrates a sound free-choice process with complex initialization and looping logic. Initially, activities $c$ and $e$ execute concurrently, and the two branches synchronize at $g$. After $g$, a decision is made to either exit or loop back. If the process loops, a choice is made between $a$ and $d$. Firing $a$ triggers the concurrent execution of $\{b, c\}$, while firing $d$ triggers $\{e, f\}$; that is, $c$ and $e$ are never re-enabled concurrently after the initial run. \cref{alg:convert_pn_to_powl_v2} attempts to construct a partial order, but the partitioning fails because the region $\{a, b, c, d, e, f, g\}$ has two non-equivalent entry points, $p_3$ and $p_4$. When attempting to construct a choice graph, the algorithm merges all core transitions into a single part, managing only to isolate the initial and final silent transitions. Because this results in a non-progressive projection that is isomorphic to the net itself ($N_i \cong N$), the choice-graph attempt also fails.
    
\end{itemize}

\section{Correctness and Completeness Guarantees}\label{sec:gua}
In this section, we prove the correctness and completeness guarantees of \cref{alg:convert_pn_to_powl_v2}. Our proof strategy for correctness relies on structural induction on the WF-net. We first show that the projection operations for both partial order and choice graph patterns preserve safeness and soundness, allowing for the recursive application of the algorithm on the created subnets (\cref{sec:lemmas:proj}).  Then, we demonstrate that combining the languages of these subnets using the respective POWL constructs accurately reflects the language of the original WF-net (\cref{sec:lemmas:pattern}). We combine these findings to prove the overall correctness of the algorithm (\cref{sec:gua:lang}). Finally, we show the completeness of our algorithm on the class of separable WF-nets (\cref{sec:gua:redisc}).

\subsection{Projection Structural Guarantees}\label{sec:lemmas:proj}
This section proves that the projections defined in \cref{def:projection_po} and \cref{def:projection_cg}, when applied to safe and sound WF-nets, produce safe and sound WF-nets. This ensures that the recursive calls in \cref{alg:convert_pn_to_powl_v2} are always applied to valid inputs.

\begin{lemma}[Partial Order Projection Structural Guarantees]\label{thm:po_soundness_preservation}
Let $N = (P, T, F)$ be a safe and sound WF-net. Let $G = \{T_1, \ldots, T_n\}$ be a conflict-hiding partition of $T$. Let $N_i = (P_i, T_i, F_i) \allowbreak = \poproject(N, T_i)$ for each i ($1 \leq i \leq n$). Then $N_i$ is a safe and sound WF-net.
\end{lemma}

\begin{proof}

\textbf{(1) $N_i$ is a WF-net:} By construction, every node in $N_i$ is on a path from $p_s$ and $p_e$. The applied normalization ensures that additional source and/or sink places are inserted in case $\pre{p_s} \neq \emptyset$ and/or $\post{p_e} \neq \emptyset$, respectively.

\textbf{(2) We prove that $N_i$ is sound.}
 
\textbf{(2.1) No dead transitions:}
    Consider any transition $t \in T_i$. Since $N$ is sound, there exists a reachable marking $M$ in $N$ that enables $t$. Consider the marking $M_i$ in $N_i$ defined as follows for $p \in P_i$:
\begin{equation*}
M_i(p) =
    \begin{cases}
        M(p) & \text{if } p \in P, \\
        1 & \text{if } p = p_s \text{ and } M(p) = 1 \text{ for all } p \in \Pre{T_i}, \\
        1 & \text{if } p = p_e \text{ and } M(p) = 1 \text{ for all } p \in \Post{T_i}, \\
        0 & \text{otherwise (for places added by normalization)} .
    \end{cases}
\end{equation*}
    The projection operation only removes places and transitions that are not in $T_i$ and replaces the connections to $\Pre{T_i}$ and $\Post{T_i}$ with $p_s$ and $p_e$, respectively. Thus, any firing sequence leading to $M$ in $N$ can be transformed into a firing sequence leading to $M_i$ in $N_i$ by removing transitions not in $T_i$ (and potentially firing the additional silent transition added by normalization). Therefore, $M_i$ is reachable from $[{N_i}_{source}]$ in $N_i$. The single entry and exit properties (cf. \cref{def:po_pattern}) ensure that if $t$ needs to consume a token from a place $p \in \Pre{T_i}$ (or $p \in \Post{T_i}$) in $N$, then all other places in $\Pre{T_i}$ (or in $\Post{T_i}$, respectively) must be included in $M$ as well. This implies that $M$ and $M_i$ agree on all places that are needed to enable $t$, including start and end places. Therefore, $t$ is enabled in $M_i$. 

\textbf{(2.2) Option to complete:}
    Consider any marking $M_i$ reachable from $[{N_i}_{source}]$ in $N_i$. Due to the single entry and exit properties (cf. \cref{def:po_pattern}), there must exist a reachable marking $M$ in $N$ that enables the transitions in $T_i$ in the same way as $M_i$ does in $N_i$, i.e., $M$ is defined for $p \in P\project{T_i}$ as follows:
\begin{equation*}
    M(p) = 
    \begin{cases}
        M_i(p_s) & \text{if } p \in \Pre{T_i}, \\
        M_i(p_e) & \text{if } p \in \Post{T_i}, \\
        M_i(p) & \text{otherwise.}
    \end{cases}
\end{equation*} 

Since $N$ is sound, there exists a firing sequence $\sigma$ from $M$ to $[N_{sink}]$ in $N$. We can construct a corresponding firing sequence $\sigma_i$ from $M_i$ to $[{N_i}_{sink}]$ in $N_i$ by taking only the transitions in $\sigma$ that belong to $T_i$ (and potentially firing the additional silent transition added by normalization).

\textbf{(3) $N_i$ is safe:}
Assume, for the sake of contradiction, that there exist a reachable marking $M_i$ in $N_i$ and a place $p \in P_i$ such that $M_i(p) \geq 2$. There must exist a reachable marking $M$ in $N$ that enables the transitions in $T_i$ in the same way as $M_i$ does in $N_i$ (cf. the proof of ``option to complete''). Then there exists $p' \in P\project{T_i}$ such that $M(p') = M_i(p) \geq 2$. This violates the safeness of $N$. \end{proof}

\begin{lemma}[Choice Graph Projection Structural Guarantees]\label{thm:cg_soundness_preservation}
Let $N = (P, T, F)$ be a safe and sound WF-net. Let $G = \{T_1, \ldots, T_n\}$ be a concurrency-hiding partition (\cref{def:cg_pattern}). Let $N_i = \cgproject(N, T_i)$ for any $T_i \in G$. Then $N_i$ is a safe and sound WF-net.
\end{lemma}

\begin{proof}

\textbf{(1) $N_i$ is a WF-net:} 
The unique connection points requirement of \cref{def:cg_pattern} guarantees that every node in $T_i$ lies on a path from the unique entry point to the unique exit point.


\textbf{(2) We prove that $N_i$ is sound.}

\textbf{(2.1) No dead transitions:}
Let $p_{in} \in \Pre{T_i}$ be the unique entry point of $T_i$ and $p_{out} \in \Post{T_i}$ be its unique exit point. Since $N$ is sound, any transition $t \in T_i$ can be enabled by some marking reachable from $N_{source}$. Since $p_{in}$ is the only entry point into the subnet $N_i$ (due to $|\Pre{T_i}| = 1$), any token reaching $t$ must have passed through $p_{in}$. Thus, in the sub-net $N_i$, $t$ is reachable from $[p_{in}]$.

\textbf{(2.2) Option to complete:}
Since $N$ has the option to complete, any token entering $T_i$ via $p_{in}$ must eventually be able to leave $T_i$ to reach $N_{sink}$. Since $p_{out}$ is the only exit point from the subnet $N_i$ (due to $|\Post{T_i}| = 1$), there must exist a firing sequence within $N_i$ leading from every reachable state to $[p_{out}]$. 

\textbf{(3) $N_i$ is safe:}
Assume $N_i$ is unsafe. Then there exists a reachable marking in $N_i$ where a place $p \in P \project{T_i}$ holds $k \ge 2$ tokens. In the original net $N$, since $G$ is a concurrency-hiding partition, the high-level structure behaves as a state machine. Specifically, no two concurrent threads can enter $T_i$ simultaneously from the outside. Therefore, any accumulation of tokens inside $T_i$ must be generated by the internal logic of $T_i$ itself. If $N_i$ generates $k \ge 2$ tokens, then $N$ would also generate these tokens when executing this part, violating the safeness of $N$. Thus, $N_i$ must be safe.
\end{proof}

\subsection{Pattern Language Preservation Guarantees}\label{sec:lemmas:pattern}
This section establishes the language preservation guarantees of the identified patterns. We prove that the generated partial order and choice graph constructs correctly capture the behavior of the original WF-net.

\begin{lemma}[Partial Order Pattern Language Preservation]\label{thm:po_lang_preservation}
Let $N = (P, T, F)$ be a safe and sound WF-net. Let $G = \{T_1, \ldots, T_n\}$ be a conflict-hiding partition of $T$ and $\po =\closure{\mathit{order}}(N, G)$. Let $\powl_1, \ldots, \powl_n$ be POWL models such that $\lang(\poproject(N, T_i)) \allowbreak = \lang(\powl_i)$ for each $i$ ($1 \leq i \leq n$). Then, $\lang(\po(\powl_1, \ldots, \powl_n)) = \lang(N)$.
\end{lemma}

\begin{proof}
Let $N_i = \poproject(N, T_i)$ for each $i$ ($1 \leq i \leq n$). By combining the semantics of partial orders (cf. \cref{def:ext_sem}) with the assumption that $\lang(\powl_i) = \lang(N_i)$ for each $i$ ($1 \leq i \leq n$), we can write:
\[
\lang(\po(\powl_1, \ldots, \powl_n)) = \{\sigma \in \shuffle_{\po}(\sigma_1 , ..., \sigma_n) \ | \ \forall_{1 \leq i \leq n} \sigma_i \in \lang(N_i)\}.
\]

\textbf{(1) Proof for $\lang(N) \subseteq \{\sigma \in \shuffle_{\po}(\sigma_1 , ..., \sigma_n) \ | \ \forall_{1 \leq i \leq n} \sigma_i \in \lang(N_i)\}$:} Let $\sigma \in \lang(N)$ be any firing sequence of $N$. We construct subsequences $\sigma_1, ..., \sigma_n$ by projecting $\sigma$ onto $T_i$ for each $i$ ($1 \leq i \leq n$). We need to show that $\sigma$ can be expressed as a shuffle of $\sigma_1, ..., \sigma_n$, respecting the partial order $\po$. This can be derived by proving the following three key points:

\begin{itemize}
    \item All parts $T_i \in G$ are present within $\sigma$.
    \item Each $\sigma_i$ is a firing sequence from $N_i$ (i.e., $\sigma_i \in \lang(N_i)$). 
    \item The partial order between the subsequences is preserved in $\sigma$  (i.e., $\sigma \in \shuffle_{\po}(\sigma_1 , ..., \sigma_n)$).
\end{itemize}

\textbf{(1.1) All parts $T_i \in G$ are present within $\sigma$:}
Assume, for the sake of contradiction, that there exists a part $T_i \in G$ such that no transitions $t \in T_i$ are fired during the execution of $\sigma$. Since $N$ is a sound WF-net, skipping a subnet requires a decision point (XOR-split) upstream that bypasses $T_i$. However, Condition 1 of \cref{def:po_pattern} states that for any place $p$ connecting parts, $|\{T_k \in G \mid p \in \Pre{T_k}\}| \leq 1$. This means flow cannot diverge to mutually exclusive parts at the top level.

\textbf{(1.2) 
Each $\sigma_i$ is a firing sequence from $N_i$:}
The single entry and exit properties (cf. \cref{def:po_pattern}) ensure each subnet is executed independently from start to end within $N$. Therefore, $\sigma_i$ must be a firing sequence from $N_i$.

\textbf{(1.3) The partial order $\po$ is preserved in $\sigma$:} 
Assume $i \po j$ for any $i, j \in \{1, \dots, n\}$. Since $\po = \closure{\mathit{order}}(N, G)$, transitions from $T_i$ are executed first, producing tokens that are needed to eventually enable $T_j$. Suppose, for the sake of contradiction, that after the execution of transitions from $T_j$, transitions from $T_i$ are re-enabled (i.e., tokens are produced in $\Pre{T_i}$). Then we have two possible scenarios:
    \begin{itemize}
        \item (i) The re-enabling of $T_i$ does not depend on the completion of $T_j$ (i.e., it does not require the consumption of tokens from $\Post{T_j}$): This means that we can perform a full execution of the subnet of $T_i$ and reach the subnet of $T_j$ again before its completion, violating safeness.
        \item (ii) The re-enabling of $T_i$ depends on the completion of $T_j$ (i.e., it requires the consumption of tokens from $\Post{T_j}$): This implies the existence of a sequence of dependencies in the execution order $j \po \dots \po i$. By transitivity, $j \po i$ holds. This violates the asymmetry requirement of partial orders since $i \po j$.
    \end{itemize}

\textbf{(2) Proof for $\{\sigma \in \shuffle_{\po}(\sigma_1 , ..., \sigma_n) \ | \ \forall_{1 \leq i \leq n} \sigma_i \in \lang(N_i)\} \subseteq \lang(N)$:}
Consider any sequence $\sigma \in \shuffle_{\po}(\sigma_1 , ..., \sigma_n)$ where $\sigma_i \in \lang(N_i)$ for $1 \leq i \leq n$. We showed that all parts $T_i \in G$ must be visited in $N$ (cf. the proof of 1.1). Due to the single entry and exit properties in $N$ (cf. \cref{def:po_pattern}), each subnet can be executed independently in $N$, without violating the execution order $\po = \closure{\mathit{order}}(N, G)$. Therefore, the interleaved sequence $\sigma$ constitutes a valid firing sequence in $N$. 

\end{proof}

\begin{lemma}[Choice Graph Pattern Language Preservation]\label{thm:cg_lang_preservation}
Let $N = (P, T, F)$ be a safe and sound WF-net. Let $G = \{T_1, \ldots, T_n\}$ be a concurrency-hiding partition of $T$ and $\cg = \mathit{flow}(N, G)$. Let $\powl_1, \ldots, \powl_n$ be POWL models such that $\lang(\cgproject(N, T_i)) \allowbreak = \lang(\powl_i)$ for each $i$ ($1 \leq i \leq n$). Then, $\lang(\cg(\powl_1, \ldots, \powl_n)) = \lang(N)$.
\end{lemma}

\begin{proof}
Let $N_i = \cgproject(N, T_i)$ for each $i \in \{1, \dots, n\}$. Based on the semantics of choice graphs (\cref{def:ext_sem}) and the assumption $\lang(\powl_i) = \lang(N_i)$, we have:
\[
\lang(\cg(\powl_1, \dots, \powl_n)) = \bigcup_{\langle p_1, \dots, p_k \rangle \in \paths{\cg}} \lang(N_{p_1}) \cdot \dots \cdot \lang(N_{p_k}).
\]

\textbf{(1) Proof for $\lang(N) \subseteq \lang(\cg(\powl_1, \dots, \powl_n))$:} 
Let $\sigma \in \lang(N)$ be a firing sequence. Since $G$ is a concurrency-hiding partition (\cref{def:cg_pattern}), every part $T_i$ acts as an SESE component and the high-level control flow between these parts in $N$ behaves as a state machine. Consequently, $\sigma$ can be uniquely partitioned into a sequence of sub-traces $\sigma = \sigma_1 \cdot \dots \cdot \sigma_k$ such that each $\sigma_j$ is a firing sequence of some subnet $N_{p_j}$ where $T_{p_j} \in G$. For each $j \in \{1, \dots, k-1\}$, the transition from the execution of $T_{p_j}$ to $T_{p_{j+1}}$ implies that a token was passed from the unique exit point $\Post{T_{p_j}}$ to the unique entry point $\Pre{T_{p_{j+1}}}$. By \cref{def:flow_graph}, this corresponds to an edge $(p_j, p_{j+1})$ in $\cg$. Similarly, the first part $T_{p_1}$ must be reachable from $\xorstart{}$ (since $N_{source} \in \Pre{T_{p_1}}$) and the last part $T_{p_k}$ must lead to $\xorend{}$ (since $N_{sink} \in \Post{T_{p_k}}$). Thus, the sequence $\langle p_1, \dots, p_k \rangle$ is a valid path in $\paths{\cg}$. 

\textbf{(2) Proof for $\lang(\cg(\powl_1, \dots, \powl_n)) \subseteq \lang(N)$:}
Let $\sigma \in \lang(\cg(\powl_1, \dots, \powl_n))$. By definition, there exists a path $\langle p_1, \dots, p_k \rangle \in \paths{\cg}$ such that $\sigma = \sigma_1 \cdot \dots \cdot \sigma_k$ with $\sigma_j \in \lang(N_{p_j})$.

We analyze the token flow in $N$ through the partition interfaces:
\begin{itemize}
    \item \textbf{Initialization:} The path starts with $(\xorstart{}, p_1)$, implying $N_{source} \in \Pre{T_{p_1}}$ (by \cref{def:flow_graph}). By \cref{def:cg_pattern}, $|\Pre{T_{p_1}}| = 1$, so $N_{source}$ is the unique entry place for $T_{p_1}$. Since $\sigma_1 \in \lang(N_{p_1})$, firing $\sigma_1$ in $N$ consumes the token from $N_{source}$ and produces a token in the unique exit place $\Post{T_{p_1}}$.
    
    \item \textbf{Chaining:} For any step $j$ ($1 \le j < k$), the edge $(p_j, p_{j+1})$ exists in $\cg$. This implies $\Post{T_{p_j}} \cap \Pre{T_{p_{j+1}}} \neq \emptyset$ (by \cref{def:flow_graph}). 
    By \cref{def:cg_pattern}, $|\Post{T_{p_j}}| = 1$ and $|\Pre{T_{p_{j+1}}}| = 1$. Therefore, the unique exit place of $T_{p_j}$ is identical to the unique entry place of $T_{p_{j+1}}$. Let this place be $p_{link}$.
    After $\sigma_j$ fires, a token resides in $p_{link}$. Since $\sigma_{j+1} \in \lang(N_{p_{j+1}})$, it is enabled by this token and fires, moving the token to $\Post{T_{p_{j+1}}}$.
    
    \item \textbf{Termination:} The path ends with $(p_k, \xorend{})$, implying $N_{sink} \in \Post{T_{p_k}}$ (by \cref{def:flow_graph}). Since $|\Post{T_{p_k}}| = 1$, the final token produced by $\sigma_k$ lands in $N_{sink}$.
\end{itemize}

Since the sequence $\sigma$ continuously transforms the marking from $[N_{source}]$ to $[N_{sink}]$ without leaving tokens behind (safeness and soundness are preserved locally by projection), $\sigma$ is a valid firing sequence in $N$. Thus, $\sigma \in \lang(N)$.
\end{proof}

\subsection{Overall Correctness Guarantee}\label{sec:gua:lang}
In this section, we prove the correctness of \cref{alg:convert_pn_to_powl_v2}. Specifically, we show that if the algorithm natively produces a POWL model (i.e., it decomposes the net without reverting to the fall-through at any level of recursion), then the generated model preserves the language of the original WF-net.

\begin{theorem}[Correctness]\label{thm:main_correctness_compact}
Let $N = (P, T, F)$ be a safe and sound WF-net. If \cref{alg:convert_pn_to_powl_v2} successfully converts $N$ into a POWL model $\powl$ without invoking the fall-through, then $\lang(N) = \lang(\powl)$.
\end{theorem}

\begin{proof}

If no pattern is detected at any level of recursion, the algorithm invokes the fall-through. In such cases, the premise of the theorem is not satisfied, and the statement holds. 

Now, assume that the fall-through is not invoked at any level of recursion. We prove the theorem by induction on the number of transitions $|T|$.

\begin{itemize}

\item \textbf{Base case:} The theorem trivially holds for a WF-net that contains a single transition.

\item \textbf{Inductive hypothesis:} For $n> 1$, assume the theorem holds for all safe and sound WF-net with fewer transitions than $n$ (i.e., $|T| < n$).

\item \textbf{Inductive step ($|T| = n$):} We distinguish two cases based on the decomposition pattern identified:

\begin{itemize}

    \item \textbf{Partial Order:} Suppose a conflict-hiding partition $G = \{T_1, \dots, T_n\}$ is identified. For each $T_i \in G$, $N$ is projected onto $T_i$ to obtain $N_i = \poproject(N, T_i)$.
     By \cref{thm:po_soundness_preservation}, each $N_i$ is a safe and sound WF-net with fewer transitions than $n$.
     By the inductive hypothesis, the POWL model $\powl_i$ obtained from $N_i$ satisfies $\lang(\powl_i) = \lang(N_i)$. The algorithm returns $\powl = \po(\powl_1, \dots, \powl_n)$ where $\po =\closure{\mathit{order}}(N, G)$.
     By \cref{thm:po_lang_preservation}, $\lang(\powl) = \lang(N)$.
    
     \item \textbf{Choice Graph:} Suppose a concurrency-hiding partition is identified. The proof is analogous to the partial order case, using \cref{thm:cg_soundness_preservation} for the structural guarantees and \cref{thm:cg_lang_preservation} for the language equivalence.

\end{itemize}

\end{itemize}

By induction, the theorem holds for all safe and sound WF-nets successfully converted by \cref{alg:convert_pn_to_powl_v2}.

\end{proof}

\subsection{Completeness Guarantee on Separable WF-Nets} \label{sec:gua:redisc}
In this section, we show that \cref{alg:convert_pn_to_powl_v2} is complete when applied to sound and safe separable WF-nets.

\begin{theorem}[Completeness]\label{thm:complete}
Let $N$ be a separable WF-net. If $N$ is safe and sound, then \cref{alg:convert_pn_to_powl_v2} successfully converts $N$ into a POWL 2.0 model without invoking the fall-through. 
\end{theorem}

\begin{proof}
We prove the theorem by induction on the number of transitions $|T|$ in the WF-net $N$.

\item \textbf{Base Case ($|T|=1$):}
If $N$ contains only one transition $t$, it must satisfy the structural requirements of a base case WF-net due to soundness. \cref{alg:convert_pn_to_powl_v2} identifies this and returns $t$. 

\item \textbf{Inductive hypothesis:} Assume the theorem holds for all separable WF-nets with fewer than $|T|$ transitions.

\item \textbf{Inductive step:} 
By \cref{def:separable}, $N$ is constructed hierarchically. Due to the associativity of substitution (e.g., substituting a state machine into a state machine results in a larger state machine), we can characterize $N$ by a \emph{flattened top-level decomposition}: either a flattened state machine where all sub-components are not state machines or a flattened marked graph where the sub-components are not marked graphs.

\cref{alg:convert_pn_to_powl_v2} attempts to apply both $\popart$ and $\cgpart$. We show that at least one succeeds by identifying the flattened top-level decomposition $G_{top} = \{T_1, \dots, T_n\}$.

\begin{itemize}
    \item \textbf{Case 1: $N$ is a top-level state machine.}
    The algorithm calls $\cgpart(N)$. We show that it constructs $G_{top}$
    \begin{itemize}
         \item \textbf{Separation (No over-merging):} The merging rules in $\cgpart$ are triggered strictly by concurrency patterns (AND-splits and AND-joins). Since the top-level structure of $N$ is a state machine, the algorithm never merges transitions across the boundaries of distinct parts in $G_{top}$.

         \item \textbf{Aggregation (No under-merging):} Each component $T_i$ is effectively a nested marked graph. The algorithm looks at all AND-splits and AND-joins from both ends and spans the threads between them. This bidirectional sweep ensures that all nodes of the marked graph are captured and merged into $T_i$. 
         The only scenario where this strategy would fail to unify $T_i$ is if the graph were sequentializable (i.e., if $T_i$ could be split into a sequence of two subgraphs $T_{i,1} \to T_{i,2}$ with no concurrency spanning them). However, this would violate the assumption that $G_{top}$ is the flattened top-level decomposition. If such a sequence existed, $T_{i,1}$ and $T_{i,2}$ would simply be two distinct, adjacent nodes in the top-level state machine, rather than forming a single component. Therefore, each $T_i$ in $G_{top}$ is fully captured by the algorithm.
    \end{itemize}
    
    Thus, $\cgpart(N)$ returns $G_{top}$. Since the top level is a state machine, the partition satisfies the requirement of \cref{def:cg_pattern} ($|\Pre{T_i}| = |\Post{T_i}| = 1$ for every $T_i$), making it a valid concurrency-hiding partition. 

    \item \textbf{Case 2: $N$ is a top-level marked graph.}
    The algorithm first calls $\cgpart(N)$. If this succeeds (e.g., the top level is a simple sequence that can be modeled as both a choice graph or a state machine), then $G_{top}$ is returned. Now, assume $\cgpart(N)$ was not successful. The algorithm calls $\popart(N)$. We show that it constructs $G_{top}$ in this case.

    \begin{itemize}
        \item \textbf{Separation (No over-merging):} The merging rules in $\popart$ are triggered by conflict patterns (XOR-splits/joins). Since the top-level structure is a marked graph, the algorithm never merges transitions across the boundaries of distinct parts in $G_{top}$.

        \item \textbf{Aggregation (No under-merging):} Each component $T_i$ is effectively a nested state machine. The algorithm looks at all XOR-splits and XOR-joins and ensures that all decision logic is pulled together from both ends. The only scenario where this strategy would fail to unify $T_i$ is if the graph were sequentializable (i.e., if $T_i$ could be split into a sequence of two subgraphs $T_{i,1} \to T_{i,2}$ with no choice spanning them). However, this would violate the assumption that $G_{top}$ is the flattened top-level decomposition. If such a sequence existed, $T_{i,1}$ and $T_{i,2}$ would simply be two distinct, adjacent nodes in the top-level marked graph, rather than forming a single component. Therefore, each $T_i$ in $G_{top}$ is fully captured by the algorithm.
        
    \end{itemize}
    Thus, $\popart(N)$ returns $G_{top}$. Since the top level is a marked graph, the connections between parts are purely for synchronization, meaning no place distributes tokens to multiple distinct parts or consumes token from multiple parts (Conditions 1 and 2 of \cref{def:po_pattern}). Furthermore, because $N$ is separable, each state machine sub-component $T_i$ must act as an SESE fragment (Conditions 3 and 4 of \cref{def:po_pattern}). Therefore, the partition is a valid conflict-hiding partition.

\end{itemize}
Since $N$ is separable, it satisfies at least one of these cases, ensuring the algorithm always find $G_{top}$. The algorithm projects $N$ onto each $T_i$. Since separability is preserved under projection, each $N_i$ is a separable WF-net with fewer transitions than $|T|$. By the inductive hypothesis, the recursive calls also succeed. 
\end{proof}

\paragraph{Role of Preprocessing in Extending Coverage:}
\cref{thm:complete} establishes that all safe and sound separable WF-nets can be successfully converted. Conversely, the successful termination of \cref{alg:convert_pn_to_powl_v2} (i.e., without invoking the fall-through at any level of recursion) inherently implies that the input WF-net must be separable. This is because the algorithm's recursive decomposition strategy fundamentally relies on identifying and applying patterns that align precisely with the definition of separable WF-nets: either a top-level marked graph or a top-level state machine. However, the preprocessing step described in \cref{sec:preproc} plays a crucial role in extending the practical applicability of our approach. Certain WF-nets, while not strictly separable in their initial form, can be transformed into equivalent separable structures. An example of this is shown in \cref{fig:pre_success}, where a non-separable net is successfully converted after preprocessing. This means that our overall approach can successfully convert a broader class of WF-nets than just those strictly separable, specifically encompassing models that our reduction rules can transform into \emph{semantically equivalent} separable nets.

\section{Implementation and Evaluation}\label{sec:eval}
To assess the scalability and robustness of our algorithm, we implemented it in Python and evaluated its performance against existing methods. The implementation includes a preprocessing layer based on the reduction rules illustrated in \cref{fig:prepro}. The code and data are available at \url{https://github.com/humam-kourani/WF-net-to-POWL}. We performed two distinct experiments.
\begin{itemize}
    \item In the first experiment, we focused on scalability using models of increasing size. We utilized the process tree generator from \cite{DBLP:conf/bpm/JouckD16,DBLP:journals/bise/JouckD19} to generate 1,000 process trees, which were then translated into WF-nets using PM4Py \cite{DBLP:journals/simpa/BertiZS23}. This resulted in a diverse set of WF-nets varying in size from $21$ to $370$ transitions and $15$ to $305$ places.
    \item In the second experiment, we assessed the completeness of the approach against real-world logic. We used a dataset of 1,493 real-world and synthetic sound workflow nets from \cite{Polyvyanyy2023}. This collection is composed of two subsets. It contains 493 industrial models sourced from the SAP R/3 Reference Model \cite{10.5555/265614}, representing diverse business domains, such as procurement, sales, and production. The remaining 1,000 models are synthetically generated to mimic the structural and behavioral complexity of the industrial set \cite{DBLP:journals/is/PolyvyanyyHROP24}. By utilizing this combined dataset, we ensure that our evaluation covers a wide spectrum of control-flow complexities, including both well-structured and unstructured real-life business logic.
\end{itemize}

For comparison, we use the process-tree converter of \cite{DBLP:journals/algorithms/ZelstL20} and our previous WF-net-to-POWL converter \cite{DBLP:conf/apn/KouraniPA25}. We additionally evaluate an RPST-based structural baseline motivated by \cite{DBLP:conf/wsfm/PolyvyanyyVV10,DBLP:conf/apn/WeidlichPMW10,DBLP:journals/fuin/WeidlichPMW11}. For each WF-net, we compute its RPST using jBPT \cite{DBLP:conf/caise/PolyvyanyyW13} and classify each rigid component as an S-net (state machine), a T-net (marked graph), or neither. We regard this baseline as successful when every rigid component is an S-net or T-net. Unlike the other three methods, this baseline performs structural decomposition and classification rather than constructing a target process model. The results are summarized in \cref{tab:summary}.

\begin{table}[!t]
\centering
\caption{Performance and success rates across both experiments.}\label{tab:summary}
\smaller
\resizebox{\textwidth}{!}{
    \begin{tabular}{llcccc}
    \hline
    \textbf{Dataset} & \textbf{Converter} & \textbf{Avg Time (s)} & \textbf{Max Time (s)} & \textbf{Num. Failures} & \textbf{Success Rate} \\ \hline
    \multirow{3}{*}{PTGen (Exp. 1)} & Process Tree \cite{DBLP:journals/algorithms/ZelstL20} & 14.65 & 125.67 &  0 / 1,000  & 100\% \\
     & POWL \cite{DBLP:conf/apn/KouraniPA25} & 0.30 & 2.48 & 0 / 1,000 & 100\% \\
     & POWL 2.0 (Proposed) & 0.22 & 1.48 &  0 / 1,000 & 100\% \\
     & RPST+S/T
       \cite{DBLP:conf/wsfm/PolyvyanyyVV10,DBLP:journals/fuin/WeidlichPMW11}
     & 0.09 & 0.67 & 0 / 1,000 & 100\% \\ \hline
    \multirow{3}{*}{SAP R/3 (Exp. 2)} & Process Tree \cite{DBLP:journals/algorithms/ZelstL20} & 0.02 & 0.54 & 112 / 1,493 & 92.5\% \\
     & POWL \cite{DBLP:conf/apn/KouraniPA25} & 0.01 & 0.36 & 27 / 1,493& 98.2\% \\
     & POWL 2.0 (Proposed) & 0.004 & 0.05 & \textbf{0} / 1,493 & \textbf{100\%} \\ 
     & RPST+S/T
       \cite{DBLP:conf/wsfm/PolyvyanyyVV10,DBLP:journals/fuin/WeidlichPMW11}
     & 0.004 & 0.05 & 36 / 1,493 & 97.6\% \\ \hline
    \end{tabular}
}
\end{table}

\begin{figure}[!t]
    \centering    
    \includegraphics[width=0.8\textwidth]{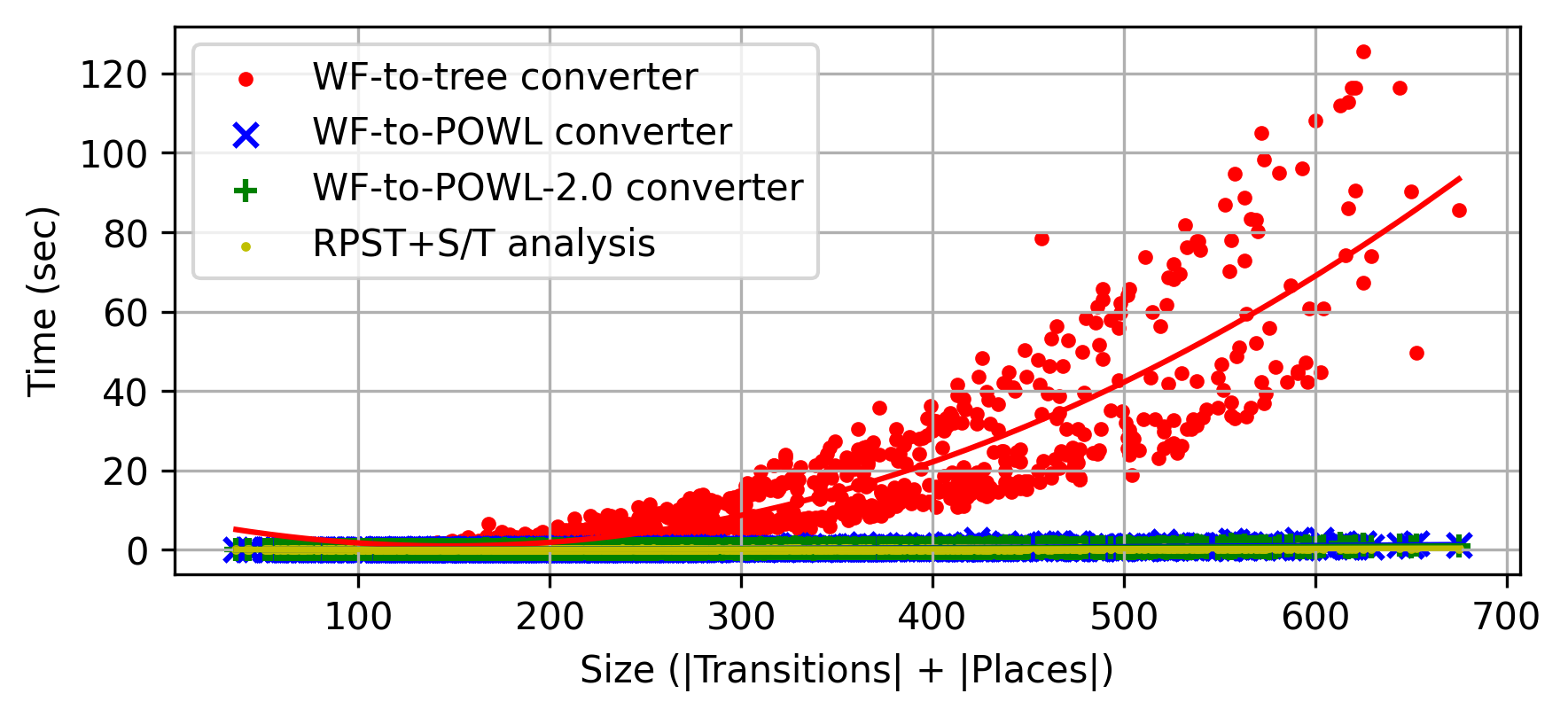}
    \caption{Comparison of runtimes on the 1,000 WF-nets of Experiment~1. For the three conversion approaches, runtime refers to end-to-end model conversion. For RPST+S/T, runtime covers RPST construction and S-/T-net classification.}
\label{fig:ev}
\end{figure}

\paragraph*{Scalability Analysis.}
In the first experiment, all three conversion algorithms successfully converted all 1,000 WF-nets, which was expected since process trees represent a subclass of POWL. The RPST analysis likewise found no rigid component in any of these models. 

The experiment demonstrated the high scalability of our approach, as illustrated in \cref{fig:ev}. While the tree-based converter from \cite{DBLP:journals/algorithms/ZelstL20} required up to $126$ seconds for the largest nets, our approaches were significantly faster. Specifically, our previous POWL algorithm required up to $2.48$ seconds for the largest models, while our new POWL 2.0 algorithm further reduced this to $1.48$ seconds. These results highlight the efficiency of our top-down decomposition strategy over the bottom-up approach used by the process tree converter. 

The RPST+S/T analysis is also highly efficient, requiring at most $0.67$ seconds. This is consistent with the fact that the RPST can be computed in linear time \cite{DBLP:conf/wsfm/PolyvyanyyVV10}. Note, however, that RPST+S/T performs only structural decomposition and classification and does not construct a target model.

\begin{figure}[!t]
    \centering    
    \begin{subfigure}{\textwidth}
        \centering
        \includegraphics[width=\textwidth]{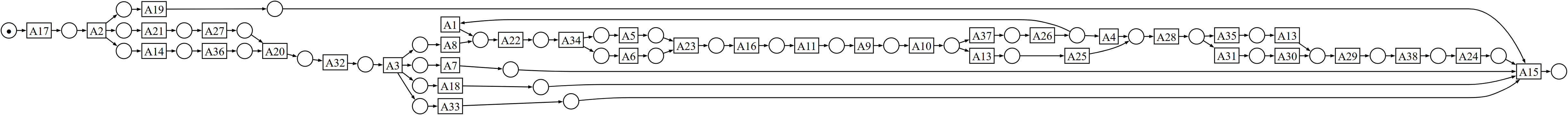}
        \caption{WF-net.}
        \label{fig:qualitative:wf}
    \end{subfigure}

    \begin{subfigure}{\textwidth}
        \centering
        \includegraphics[width=0.69\textwidth]{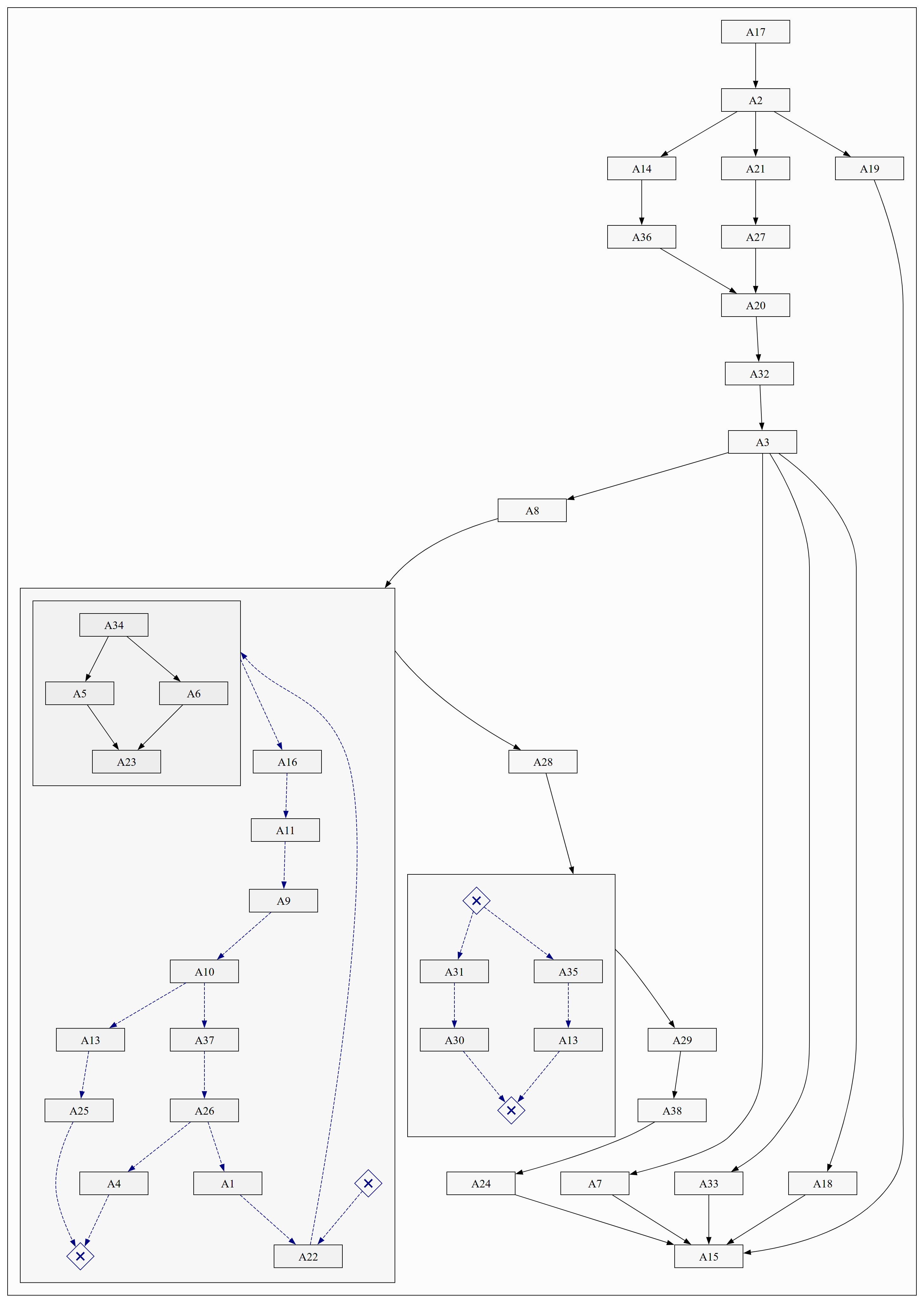}
        \caption{Converted POWL 2.0 model.}
        \label{fig:qualitative:powl}
    \end{subfigure}
    \caption{An example from the SAP R/3 dataset (process id: 48\_177) where previous approaches fail. For compactness, the original activity labels are mapped to $A_1, A_2, \dots$.\label{fig:qualitative}}
\end{figure}

\paragraph*{Practical Expressiveness.}
The second experiment highlights the expressive power of POWL 2.0. Our proposed approach successfully converted all 1,493 industrial and synthetic models. The process tree and the previous POWL implementation converters failed on $112$ and $27$ models, respectively. As illustrated by the example in \cref{fig:qualitative}, these failures typically occur when a process exhibits non-block-structured decision logic or complex cyclic jumps that cannot be captured by the rigid operators of process trees or the previous version of POWL. Notably, while POWL 2.0 represents a formal subclass of sound WF-nets, the 100\% success rate on this dataset suggests that real-world business processes rarely exceed the complexity covered by the language. However, we acknowledge that this insight is tied to the specific dataset used and may not generalize across all domains.

The RPST analysis provides additional insight into these results. Overall, 112 models contain rigid RPST components, with 117 rigid components in total. Of these, 25 are S-nets and 56 are T-nets. Hence, RPST decomposition followed by S-/T-net classification covers 1,457 of the 1,493 models ($97.6\%$). The remaining 36 models contain a rigid component that is neither an S-net nor a T-net.

We further analyzed these 36 components in relation to the maximal-structuring approach of \cite{DBLP:journals/cj/PolyvyanyyGFW14}. Eleven are cyclic and therefore outside the scope of that acyclic approach. The remaining 25 are acyclic, and the ordering-relations/modular-decomposition
analysis identifies a concurrent primitive in all 25 cases. For such components, maximal structuring may expose additional nested structure, but a residual rigid component remains.

In summary, our experiments demonstrate the advantage of the proposed approach in both scalability and structural coverage. Note that the implemented algorithm is also available in ProMoAI \cite{DBLP:conf/ijcai/KouraniB0A24} (\url{https://promoai.streamlit.app/}), powering the redesign feature for improving existing process models via large language models.

\section{Conclusion}\label{sec:conc}
In this paper, we presented a recursive algorithm for transforming safe and sound workflow nets into equivalent POWL 2.0 models. By leveraging the unified structural power of \emph{choice graphs} and \emph{partial orders}, the algorithm overcomes the limitations of previous block-structured approaches. Our method employs a top-down decomposition strategy that identifies conflict-hiding and concurrency-hiding partitions, allowing it to handle complex, non-block-structured decisions and cyclic logic that were previously untranslatable.

The main contribution of this work lies in bridging the gap between the theoretical advantages of the hierarchical POWL language, such as guaranteed soundness, and the practical reality where Petri nets and BPMN remain the standard. We provided formal proofs of correctness, ensuring that the generated POWL models preserve the language of the original WF-nets. Furthermore, we demonstrated that the algorithm is complete for the class of \emph{separable workflow nets}, a class that encompasses the expressive scope of POWL 2.0.

Our experimental evaluation confirmed the practical utility of the approach. The algorithm exhibited superior scalability on large process models. Moreover, the approach achieved a 100\% success rate on a large dataset of industrial business processes, demonstrating that the expressive scope of POWL 2.0 is well-suited for practical, real-world applications.

This work paves the way for broader adoption of POWL across various process mining applications. A key avenue for future research is the development of optimized analysis techniques that exploit the specific structural properties of POWL, such as efficient conformance checking algorithms similar to those developed for process trees \cite{DBLP:conf/bpm/RochaA23}. Furthermore, we plan to investigate potential extensions to the POWL language to support more complex structures beyond separable workflow nets. Finally, a critical area for advancement lies in enhancing the usability of these models through interactive visualizations that overlay performance and conformance perspectives directly onto the model structure, thereby providing a comprehensive and intuitive view of process behavior.

\textbf{Acknowledgments.} 
This work was funded by the Federal Ministry of Research, Technology and Space (BMFTR), Germany (grant 16IS23065).


\bibliographystyle{fundam}
\bibliography{references}

@inproceedings{DBLP:conf/bpmds/KouraniB0A24,
  author       = {Humam Kourani and
                  Alessandro Berti and
                  Daniel Schuster and
                  Wil M. P. van der Aalst},
  title        = {Process Modeling with Large Language Models},
  booktitle    = {Enterprise, Business-Process and Information Systems Modeling - 25th
                  International Conference, {BPMDS} 2024, and 29th International Conference,
                  {EMMSAD} 2024, Limassol, Cyprus, June 3-4, 2024, Proceedings},
  pages        = {229--244},
  year         = {2024},
  bibcrossref     = {DBLP:conf/bpmds/2024},
  doi          = {10.1007/978-3-031-61007-3\_18},
  bibsource    = {dblp computer science bibliography, https://dblp.org}
}

@inproceedings{DBLP:conf/bpm/KouraniZ23,
  author       = {Humam Kourani and
                  Sebastiaan J. van Zelst},
  title        = {{POWL:} Partially Ordered Workflow Language},
  booktitle    = {Business Process Management - 21st International Conference, {BPM}
                  2023, Utrecht, The Netherlands, September 11-15, 2023, Proceedings},
  pages        = {92--108},
  year         = {2023},
  bibcrossref     = {DBLP:conf/bpm/2023},
  doi          = {10.1007/978-3-031-41620-0\_6},
  bibsource    = {dblp computer science bibliography, https://dblp.org}
}

@inproceedings{DBLP:conf/ijcai/KouraniB0A24,
  author       = {Humam Kourani and
                  Alessandro Berti and
                  Daniel Schuster and
                  Wil M. P. van der Aalst},
  title        = {{ProMoAI}: Process Modeling with Generative {AI}},
  booktitle    = {Proceedings of the Thirty-Third International Joint Conference on
                  Artificial Intelligence, {IJCAI} 2024, Jeju, South Korea, August 3-9,
                  2024},
  pages        = {8708--8712},
  year         = {2024},
  bibcrossref     = {DBLP:conf/ijcai/2024},
  bibsource    = {dblp computer science bibliography, https://dblp.org}
}

@book{DBLP:series/lnbip/Leemans22,
  author       = {Sander J. J. Leemans},
  title        = {Robust Process Mining with Guarantees - Process Discovery, Conformance
                  Checking and Enhancement},
  series       = {Lecture Notes in Business Information Processing},
  volume       = {440},
  publisher    = {Springer},
  year         = {2022},
  doi          = {10.1007/978-3-030-96655-3},
  isbn         = {978-3-030-96654-6},
  bibsource    = {dblp computer science bibliography, https://dblp.org}
}

@article{DBLP:journals/simpa/BertiZS23,
  author       = {Alessandro Berti and
                  Sebastiaan J. van Zelst and
                  Daniel Schuster},
  title        = {{PM4Py}: {A} process mining library for Python},
  journal      = {Softw. Impacts},
  volume       = {17},
  pages        = {100556},
  year         = {2023},
  doi          = {10.1016/J.SIMPA.2023.100556},
  bibsource    = {dblp computer science bibliography, https://dblp.org}
}

@article{DBLP:journals/is/KouraniZSA25,
  author       = {Humam Kourani and
                  Sebastiaan J. van Zelst and
                  Daniel Schuster and
                  Wil M. P. van der Aalst},
  title        = {Discovering partially ordered workflow models},
  journal      = {Inf. Syst.},
  volume       = {128},
  pages        = {102493},
  year         = {2025},
  doi          = {10.1016/J.IS.2024.102493},
  bibsource    = {dblp computer science bibliography, https://dblp.org}
}

@incollection{DBLP:books/el/15/RosingWCM15,
  author       = {Mark von Rosing and
                  Stephen White and
                  Fred Cummins and
                  Henk de Man},
  editor       = {Mark von Rosing and
                  Henrik von Scheel and
                  August{-}Wilhelm Scheer},
  title        = {Business Process Model and Notation - {BPMN}},
  booktitle    = {The Complete Business Process Handbook: Body of Knowledge from Process
                  Modeling to BPM, Volume {I}},
  pages        = {429--453},
  publisher    = {Morgan Kaufmann/Elsevier},
  address      = {Massachusetts, USA},
  year         = {2015},
  doi          = {10.1016/B978-0-12-799959-3.00021-5},
  bibsource    = {dblp computer science bibliography, https://dblp.org}
}

@inproceedings{DBLP:conf/apn/KouraniPA25,
  author       = {Humam Kourani and
                  Gyunam Park and
                  Wil M. P. van der Aalst},
  editor       = {Elvio Gilberto Amparore and
                  Lukasz Mikulski},
  title        = {Translating Workflow Nets into the Partially Ordered Workflow Language},
  booktitle    = {Application and Theory of Petri Nets and Concurrency - 46th International
                  Conference, {PETRI} {NETS} 2025, Paris, France, June 22-27, 2025,
                  Proceedings},
  series       = {Lecture Notes in Computer Science},
  volume       = {15714},
  pages        = {242--264},
  publisher    = {Springer},
  year         = {2025},
  doi          = {10.1007/978-3-031-94634-9\_12},
  bibsource    = {dblp computer science bibliography, https://dblp.org}
}

@inproceedings{DBLP:conf/bpm/KouraniPA25,
  author       = {Humam Kourani and
                  Gyunam Park and
                  Wil M. P. van der Aalst},
  editor       = {Arik Senderovich and
                  Cristina Cabanillas and
                  Irene Vanderfeesten and
                  Hajo A. Reijers},
  title        = {Unlocking Non-Block-Structured Decisions: Inductive Mining with Choice
                  Graphs},
  booktitle    = {Business Process Management - 23rd International Conference, {BPM}
                  2025, Seville, Spain, August 31 - September 5, 2025, Proceedings},
  series       = {Lecture Notes in Computer Science},
  volume       = {16044},
  pages        = {144--161},
  publisher    = {Springer},
  year         = {2025},
  doi          = {10.1007/978-3-032-02867-9\_10},
  bibsource    = {dblp computer science bibliography, https://dblp.org}
}

@inproceedings{DBLP:conf/bpm/JouckD16,
  author       = {Toon Jouck and
                  Beno{\^{\i}}t Depaire},
  editor       = {Leonardo Azevedo and
                  Cristina Cabanillas},
  title        = {{PTandLogGenerator}: {A} Generator for Artificial Event Data},
  booktitle    = {Proceedings of the {BPM} Demo Track 2016 Co-located with the 14th
                  International Conference on Business Process Management {(BPM} 2016),
                  Rio de Janeiro, Brazil, September 21, 2016},
  series       = {{CEUR} Workshop Proceedings},
  volume       = {1789},
  pages        = {23--27},
  publisher    = {CEUR-WS.org},
  year         = {2016},
  bibsource    = {dblp computer science bibliography, https://dblp.org}
}

@article{DBLP:journals/bise/JouckD19,
  author       = {Toon Jouck and
                  Beno{\^{\i}}t Depaire},
  title        = {Generating Artificial Data for Empirical Analysis of Control-flow
                  Discovery Algorithms - {A} Process Tree and Log Generator},
  journal      = {Bus. Inf. Syst. Eng.},
  volume       = {61},
  number       = {6},
  pages        = {695--712},
  year         = {2019},
  doi          = {10.1007/S12599-018-0541-5},
  bibsource    = {dblp computer science bibliography, https://dblp.org}
}

@article{DBLP:journals/is/FavreFV15,
  author       = {C{\'{e}}dric Favre and
                  Dirk Fahland and
                  Hagen V{\"{o}}lzer},
  title        = {The relationship between workflow graphs and free-choice workflow
                  nets},
  journal      = {Inf. Syst.},
  volume       = {47},
  pages        = {197--219},
  year         = {2015},
  doi          = {10.1016/J.IS.2013.12.004},
  bibsource    = {dblp computer science bibliography, https://dblp.org}
}

@book{desel1995free,
  title={Free choice Petri nets},
  author={Desel, Jorg and Esparza, Javier},
  number={40},
  year={1995},
  publisher={Cambridge university press}
}

@proceedings{DBLP:conf/ac/1996petri1,
  editor       = {Wolfgang Reisig and
                  Grzegorz Rozenberg},
  title        = {Lectures on Petri Nets {I:} Basic Models, Advances in Petri Nets,
                  the volumes are based on the Advanced Course on Petri Nets, held in
                  Dagstuhl, September 1996},
  series       = {Lecture Notes in Computer Science},
  volume       = {1491},
  publisher    = {Springer},
  year         = {1998},
  doi          = {10.1007/3-540-65306-6},
  isbn         = {3-540-65306-6},
  bibsource    = {dblp computer science bibliography, https://dblp.org}
}

@inproceedings{DBLP:conf/ac/Berthelot86,
  author       = {G{\'{e}}rard Berthelot},
  editor       = {Wilfried Brauer and
                  Wolfgang Reisig and
                  Grzegorz Rozenberg},
  title        = {Transformations and Decompositions of Nets},
  booktitle    = {Petri Nets: Central Models and Their Properties, Advances in Petri
                  Nets 1986, Part I, Proceedings of an Advanced Course, Bad Honnef,
                  Germany, 8-19 September 1986},
  series       = {Lecture Notes in Computer Science},
  volume       = {254},
  pages        = {359--376},
  publisher    = {Springer},
  year         = {1986},
  doi          = {10.1007/BFB0046845},
  bibsource    = {dblp computer science bibliography, https://dblp.org}
}

@InProceedings{10.1007/BFb0016204,
    author="Berthelot, G{\'{e}}rard
    and Lri-Iie",
    editor="Rozenberg, G.",
    title="Checking properties of nets using transformations",
    booktitle="Advances in Petri Nets 1985",
    year="1986",
    publisher="Springer Berlin Heidelberg",
    address="Berlin, Heidelberg",
    pages="19--40",
    isbn="978-3-540-39822-6"
}

@article{DBLP:journals/pieee/Murata89,
  author       = {Tadao Murata},
  title        = {Petri nets: Properties, analysis and applications},
  journal      = {Proc. {IEEE}},
  volume       = {77},
  number       = {4},
  pages        = {541--580},
  year         = {1989},
  doi          = {10.1109/5.24143},
  bibsource    = {dblp computer science bibliography, https://dblp.org}
}

@article{gardner2003uml,
  title={{UML} modelling of automated business processes with a mapping to {BPEL4WS}},
  author={Gardner, Tracy},
  journal={Orientation and Web Services},
  volume={30},
  year={2003}
}

@inproceedings{DBLP:conf/apn/LangnerSW98,
  author       = {Peter Langner and
                  Christoph Schneider and
                  Joachim Wehler},
  editor       = {J{\"{o}}rg Desel and
                  Manuel Silva Su{\'{a}}rez},
  title        = {Petri Net Based Certification of Event-Driven Process Chains},
  booktitle    = {Application and Theory of Petri Nets 1998, 19th International Conference,
                  {ICATPN} '98, Lisbon, Portugal, June 22-26, 1998, Proceedings},
  series       = {Lecture Notes in Computer Science},
  volume       = {1420},
  pages        = {286--305},
  publisher    = {Springer},
  year         = {1998},
  doi          = {10.1007/3-540-69108-1\_16},
  bibsource    = {dblp computer science bibliography, https://dblp.org}
}

@inproceedings{DBLP:conf/apn/Dal-Zilio20,
  author       = {Silvano Dal{-}Zilio},
  editor       = {Ryszard Janicki and
                  Natalia Sidorova and
                  Thomas Chatain},
  title        = {{MCC:} {A} Tool for Unfolding Colored Petri Nets in {PNML} Format},
  booktitle    = {Application and Theory of Petri Nets and Concurrency - 41st International
                  Conference, {PETRI} {NETS} 2020, Paris, France, June 24-25, 2020,
                  Proceedings},
  series       = {Lecture Notes in Computer Science},
  volume       = {12152},
  pages        = {426--435},
  publisher    = {Springer},
  year         = {2020},
  doi          = {10.1007/978-3-030-51831-8\_23},
  bibsource    = {dblp computer science bibliography, https://dblp.org}
}

@inproceedings{DBLP:conf/apn/Aalst21,
  author       = {Wil M. P. van der Aalst},
  editor       = {Didier Buchs and
                  Josep Carmona},
  title        = {Reduction Using Induced Subnets to Systematically Prove Properties
                  for Free-Choice Nets},
  booktitle    = {Application and Theory of Petri Nets and Concurrency - 42nd International
                  Conference, {PETRI} {NETS} 2021, Virtual Event, June 23-25, 2021,
                  Proceedings},
  series       = {Lecture Notes in Computer Science},
  volume       = {12734},
  pages        = {208--229},
  publisher    = {Springer},
  year         = {2021},
  doi          = {10.1007/978-3-030-76983-3\_11},
  bibsource    = {dblp computer science bibliography, https://dblp.org}
}

@article{DBLP:journals/algorithms/ZelstL20,
  author       = {Sebastiaan J. van Zelst and
                  Sander J. J. Leemans},
  title        = {Translating Workflow Nets to Process Trees: An Algorithmic Approach},
  journal      = {Algorithms},
  volume       = {13},
  number       = {11},
  pages        = {279},
  year         = {2020},
  doi          = {10.3390/A13110279},
  bibsource    = {dblp computer science bibliography, https://dblp.org}
}

@article{DBLP:journals/infsof/AalstL08,
  author       = {Wil M. P. van der Aalst and
                  Kristian Bisgaard Lassen},
  title        = {Translating unstructured workflow processes to readable {BPEL:} Theory
                  and implementation},
  journal      = {Inf. Softw. Technol.},
  volume       = {50},
  number       = {3},
  pages        = {131--159},
  year         = {2008},
  doi          = {10.1016/J.INFSOF.2006.11.004},
  bibsource    = {dblp computer science bibliography, https://dblp.org}
}

@inproceedings{DBLP:conf/otm/LassenA06,
  author       = {Kristian Bisgaard Lassen and
                  Wil M. P. van der Aalst},
  editor       = {Robert Meersman and
                  Zahir Tari},
  title        = {{WorkflowNet2BPEL4WS}: {A} Tool for Translating Unstructured Workflow
                  Processes to Readable {BPEL}},
  booktitle    = {On the Move to Meaningful Internet Systems 2006: CoopIS, DOA, GADA,
                  and ODBASE, {OTM} Confederated International Conferences, CoopIS,
                  DOA, GADA, and {ODBASE} 2006, Montpellier, France, October 29 - November
                  3, 2006. Proceedings, Part {I}},
  series       = {Lecture Notes in Computer Science},
  volume       = {4275},
  pages        = {127--144},
  publisher    = {Springer},
  year         = {2006},
  doi          = {10.1007/11914853\_9},
  bibsource    = {dblp computer science bibliography, https://dblp.org}
}

@article{DBLP:journals/infsof/DijkmanDO08,
  author       = {Remco M. Dijkman and
                  Marlon Dumas and
                  Chun Ouyang},
  title        = {Semantics and analysis of business process models in {BPMN}},
  journal      = {Inf. Softw. Technol.},
  volume       = {50},
  number       = {12},
  pages        = {1281--1294},
  year         = {2008},
  doi          = {10.1016/J.INFSOF.2008.02.006},
  bibsource    = {dblp computer science bibliography, https://dblp.org}
}

@article{Polyvyanyy2023,
    author = "Artem Polyvyanyy",
    title = "{Sound Workflow Systems Used in PQL Evaluations}",
    year = "2023",
    month = "1",
    opturl = "https://figshare.unimelb.edu.au/articles/dataset/Sound_Workflow_Systems_Used_in_PQL_Evaluations/21937259",
    doi = "10.26188/21937259.v1"
}

@book{10.5555/265614,
    author = {Curran, Thomas and Keller, Gerhard and Ladd, Andrew},
    title = {SAP R/3 business blueprint: understanding the business process reference model},
    year = {1997},
    isbn = {0135211476},
    publisher = {Prentice-Hall, Inc.},
    address = {USA}
}

@article{DBLP:journals/is/PolyvyanyyHROP24,
  author       = {Artem Polyvyanyy and
                  Arthur H. M. ter Hofstede and
                  Marcello La Rosa and
                  Chun Ouyang and
                  Anastasiia Pika},
  title        = {Process Query Language: Design, Implementation, and Evaluation},
  journal      = {Inf. Syst.},
  volume       = {122},
  pages        = {102337},
  year         = {2024},
  doi          = {10.1016/J.IS.2023.102337},
  bibsource    = {dblp computer science bibliography, https://dblp.org}
}

@inproceedings{DBLP:conf/bpm/RochaA23,
  author       = {Eduardo Goulart Rocha and
                  Wil M. P. van der Aalst},
  editor       = {Chiara Di Francescomarino and
                  Andrea Burattin and
                  Christian Janiesch and
                  Shazia Sadiq},
  title        = {Polynomial-Time Conformance Checking for Process Trees},
  booktitle    = {Business Process Management - 21st International Conference, {BPM}
                  2023, Utrecht, The Netherlands, September 11-15, 2023, Proceedings},
  series       = {Lecture Notes in Computer Science},
  volume       = {14159},
  pages        = {109--125},
  publisher    = {Springer},
  year         = {2023},
  opturl          = {https://doi.org/10.1007/978-3-031-41620-0\_7},
  doi          = {10.1007/978-3-031-41620-0\_7},
  bibsource    = {dblp computer science bibliography, https://dblp.org}
}

@article{DBLP:journals/cj/PolyvyanyyGFW14,
  author       = {Artem Polyvyanyy and
                  Luciano Garc{\'{\i}}a{-}Ba{\~{n}}uelos and
                  Dirk Fahland and
                  Mathias Weske},
  title        = {Maximal Structuring of Acyclic Process Models},
  journal      = {Comput. J.},
  volume       = {57},
  number       = {1},
  pages        = {12--35},
  year         = {2014},
  doi          = {10.1093/COMJNL/BXS126},
  bibsource    = {dblp computer science bibliography, https://dblp.org}
}

@inproceedings{DBLP:conf/apn/EsparzaS89,
  author       = {Javier Esparza and
                  Manuel Silva Su{\'{a}}rez},
  editor       = {Grzegorz Rozenberg},
  title        = {Circuits, handles, bridges and nets},
  booktitle    = {Advances in Petri Nets 1990 [10th International Conference on Applications
                  and Theory of Petri Nets, Bonn, Germany, June 1989, Proceedings]},
  series       = {Lecture Notes in Computer Science},
  volume       = {483},
  pages        = {210--242},
  publisher    = {Springer},
  year         = {1989},
  doi          = {10.1007/3-540-53863-1\_27},
  bibsource    = {dblp computer science bibliography, https://dblp.org}
}

@inproceedings{DBLP:conf/apn/Aalst97,
  author       = {Wil M. P. van der Aalst},
  editor       = {Pierre Az{\'{e}}ma and
                  Gianfranco Balbo},
  title        = {Verification of Workflow Nets},
  booktitle    = {Application and Theory of Petri Nets 1997, 18th International Conference,
                  {ICATPN} '97, Toulouse, France, June 23-27, 1997, Proceedings},
  series       = {Lecture Notes in Computer Science},
  volume       = {1248},
  pages        = {407--426},
  publisher    = {Springer},
  year         = {1997},
  opturl          = {https://doi.org/10.1007/3-540-63139-9\_48},
  doi          = {10.1007/3-540-63139-9\_48},
  bibsource    = {dblp computer science bibliography, https://dblp.org}
}

@inproceedings{DBLP:conf/apn/WeidlichPMW10,
  author       = {Matthias Weidlich and
                  Artem Polyvyanyy and
                  Jan Mendling and
                  Mathias Weske},
  editor       = {Johan Lilius and
                  Wojciech Penczek},
  title        = {Efficient Computation of Causal Behavioural Profiles Using Structural
                  Decomposition},
  booktitle    = {Applications and Theory of Petri Nets, 31st International Conference,
                  {PETRI} {NETS} 2010, Braga, Portugal, June 21-25, 2010. Proceedings},
  series       = {Lecture Notes in Computer Science},
  volume       = {6128},
  pages        = {63--83},
  publisher    = {Springer},
  year         = {2010},
  opturl          = {https://doi.org/10.1007/978-3-642-13675-7\_6},
  doi          = {10.1007/978-3-642-13675-7\_6},
  bibsource    = {dblp computer science bibliography, https://dblp.org}
}

@article{DBLP:journals/fuin/WeidlichPMW11,
  author       = {Matthias Weidlich and
                  Artem Polyvyanyy and
                  Jan Mendling and
                  Mathias Weske},
  title        = {Causal Behavioural Profiles - Efficient Computation, Applications,
                  and Evaluation},
  journal      = {Fundam. Informaticae},
  volume       = {113},
  number       = {3-4},
  pages        = {399--435},
  year         = {2011},
  opturl          = {https://doi.org/10.3233/FI-2011-614},
  doi          = {10.3233/FI-2011-614},
  bibsource    = {dblp computer science bibliography, https://dblp.org}
}

@phdthesis{DBLP:phd/dnb/Polyvyanyy12,
  author       = {Artem Polyvyanyy},
  title        = {Structuring process models},
  school       = {University of Potsdam},
  year         = {2012},
  url          = {http://opus.kobv.de/ubp/volltexte/2012/5902/},
  urn          = {urn:nbn:de:kobv:517-opus-59024},
  bibsource    = {dblp computer science bibliography, https://dblp.org}
}

@inproceedings{DBLP:conf/wsfm/PolyvyanyyVV10,
  author       = {Artem Polyvyanyy and
                  Jussi Vanhatalo and
                  Hagen V{\"{o}}lzer},
  editor       = {Mario Bravetti and
                  Tevfik Bultan},
  title        = {Simplified Computation and Generalization of the Refined Process Structure
                  Tree},
  booktitle    = {Web Services and Formal Methods - 7th International Workshop, {WS-FM}
                  2010, Hoboken, NJ, USA, September 16-17, 2010. Revised Selected Papers},
  series       = {Lecture Notes in Computer Science},
  volume       = {6551},
  pages        = {25--41},
  publisher    = {Springer},
  year         = {2010},
  opturl          = {https://doi.org/10.1007/978-3-642-19589-1\_2},
  doi          = {10.1007/978-3-642-19589-1\_2},
  bibsource    = {dblp computer science bibliography, https://dblp.org}
}

@inproceedings{DBLP:conf/caise/PolyvyanyyW13,
  author       = {Artem Polyvyanyy and
                  Matthias Weidlich},
  editor       = {R{\'{e}}becca Deneck{\`{e}}re and
                  Henderik Alex Proper},
  title        = {Towards a Compendium of Process Technologies - The jBPT Library for
                  Process Model Analysis},
  booktitle    = {Proceedings of the CAiSE'13 Forum at the 25th International Conference
                  on Advanced Information Systems Engineering (CAiSE), Valencia, Spain,
                  June 20th, 2013},
  series       = {{CEUR} Workshop Proceedings},
  volume       = {998},
  pages        = {106--113},
  publisher    = {CEUR-WS.org},
  year         = {2013},
  url          = {https://ceur-ws.org/Vol-998/Paper14.pdf},
  bibsource    = {dblp computer science bibliography, https://dblp.org}
}

@article{DBLP:journals/acta/PolyvyanyyWW11,
  author       = {Artem Polyvyanyy and
                  Matthias Weidlich and
                  Mathias Weske},
  title        = {Connectivity of workflow nets: the foundations of stepwise verification},
  journal      = {Acta Informatica},
  volume       = {48},
  number       = {4},
  pages        = {213--242},
  year         = {2011},
  opturl          = {https://doi.org/10.1007/s00236-011-0137-8},
  doi          = {10.1007/S00236-011-0137-8},
  bibsource    = {dblp computer science bibliography, https://dblp.org}
}

\end{document}